\documentclass[final,1p,12pt]{elsarticle}

\usepackage{graphicx}
\graphicspath{{Plots/}}
\usepackage{amssymb}
\usepackage{url}
\usepackage[dvipsnames]{xcolor}
\usepackage{ntheorem}
\usepackage[fleqn]{amsmath}
\usepackage{tablists}
    \restorelistitem 
\usepackage{bm}
\newtheorem{theorem}{Theorem}
\newtheorem{hyp}{Hypothesis}
\newtheorem{lemma}{Lemma}%[section]
\newtheorem{theorem1}{Theorem A}%[section]
\theorembodyfont{\upshape}
\theoremseparator{.}
\newtheorem*{remark}{Remark 1}
\newtheorem*{proof}{Proof of Lemma 1}
\newtheorem*{proof1}{Proof of Theorem 1}
\newtheorem*{proof2}{Proof of Theorem 2}
\usepackage{lineno}
\usepackage[colorinlistoftodos]{todonotes}
\usepackage{algorithm}
\usepackage{natbib}

\journal{Transportation Research Part C}

\begin{document}

\begin{frontmatter}


\title{Managing network congestion with a tradable credit scheme: a trip-based MFD approach}

%%\author{Name\fnref{label22}}
%%\fntext[label22]{test}
%%\cortext[cor1]{}


\author[1]{Renming Liu\corref{corr1}}
\author[2]{Siyu Chen}
\author[1]{Yu Jiang\corref{corr2}}
\author[3]{Ravi Seshadri}
\author[2]{Moshe E. Ben-Akiva}
\author[1]{Carlos Lima Azevedo}

\address[1]{DTU Management, Technical University of Denmark, Denmark}
\address[2]{CEE, Massachusetts Institute of Technology, United States}
\address[3]{Singapore-MIT Alliance for Research and Technology, Singapore}
\cortext[corr1]{corresponding author}
\cortext[corr2]{co-corresponding author}

\begin{abstract}
This study investigates the efficiency and effectiveness of an area-based tradable credit scheme (TCS) using the trip-based Macroscopic Fundamental Diagram model for the morning commute problem. In the proposed TCS, the regulator distributes initial credits to all travelers and designs a time-varying and trip length specific credit tariff. Credits are traded between travelers and the regulator via a credit market, and the credit price is determined by the demand and supply of credits. The heterogeneity of travelers is considered in terms of desired arrival time, trip length and departure-time choice preferences. The TCS is incorporated into a day-to-day modelling framework to examine the travelers' learning process, the evolution of network, and the properties of the credit market. The existence of an equilibrium solution and the uniqueness of the credit price at the equilibrium state are established analytically. Furthermore, an open-source simulation framework is developed to validate the analytical properties of the proposed TCS and compare it with alternative control strategies in terms of mobility, network performance, and social welfare. Bayesian optimization is then adopted to optimize the credit toll scheme. The numerical results demonstrate that the proposed TCS outperforms the no-control case and matches the performance of the time-of-day pricing strategy, while maintaining revenue-neutral nature.

\end{abstract}


\end{frontmatter}

%\linenumbers

\section{Introduction}
\label{S:1}
Road traffic externalities are a serious problem that affect urban transportation networks worldwide and their severity continues to increase imposing significant costs on the traveler, environment, economy, and society. Efforts to alleviate the externalities can be explored from either the supply or the demand perspective. While traditional solutions on the supply side based on building additional infrastructure are known to sometimes be counterproductive \citep{johnston1995capacity}, demand management solutions, from the widely used price instruments to the emerging, but less explored, quantity control instrument attract more attention. 

Since the profound work by \citet{pigou2013economics}, congestion pricing (CP) has received a great deal of attention over a century in both theory and practice due to its potential gains in social welfare \citep{lindsey2006economists}. Nevertheless, road pricing often receives political and social resistance as it is perceived as a tax \citep{de2020tradable}. For this reason, researchers have been exploring alternative and more appealing demand management solutions such as the tradable credit scheme (TCS) in recent years \citep{fan2013tradable, grant2014role, dogterom2017tradable}. A typical TCS system has the following features \citep{fan2013tradable}: 1) a total quota of credits available for the study area is prespecified; 2) a regulator provides an initial endowment of credits to all potential travelers; 3) the credits can be bought and sold in a market that is monitored by the regulator at a price determined by demand and supply interactions; 4) in order to travel, travelers need to spend a certain number of credits (i.e., tariff) to access the urban transportation system. The credits, also termed as permits in this study, could vary with the conditions of the specific mobility alternative used; 5) enforcement is necessary to ensure the permits being consumed or traded validly. Consequently, the TCS is a revenue-neutral and equitable demand management scheme with potentially high public acceptance \citep{verhoef1997tradeable,de2018congestion}.

Recently, several researchers have studied the modeling of the TCS considering a static traffic equilibrium in terms of flow pattern and credit market price. For example, \citet{yang2011managing} demonstrated that for a given link-specific credit charge scheme (i.e., a credit toll vector), the equilibrium link flow pattern is unique under standard assumptions, and the credit price is also unique under an additional mild assumption (i.e., there are at least two paths with different credit tariffs connecting the same origin-destination pair in the equilibrium path set). \citet{xiao2013managing} showed that the optimal time-varying charge of credits at a bottleneck always exists under the assumption that late arrival is prohibited. \citet{miralinaghi2016multi} developed a multi-period TCS to guarantee a stable credit price. The numerical results for a general network demonstrated that by setting a penalty on transferring credits from one period to a future period and a credit recycling price, the credit price becomes less volatile. \citet{de2018congestion} compared a standard TCS with a congestion pricing in terms of social welfare, using a network containing two travel modes, i.e., highways with route-specific tolls (in either dollars or credits) and public transport with a zero toll. The analytical analysis and numerical results showed that the TCS is equivalent to the congestion pricing under fully adaptive tolls, but outperforms it under non-adaptive tolls typically when the congestion function is relatively steep compared to the demand function. \citet{bao2019regulating} clarified that the equilibrium 
related to the credit price and departure rate of the bottleneck model with a TCS is not unique and proved the uniqueness using an alternative congestion model developed by \citet{CHU1995324}, where travel time only relates to the arrival rate when the trip is completed.

Although there is a large body of research that utilizes static network equilibrium approaches, few studies investigate the dynamics of the credit price considering traveler behavior at the individual level. \citet{ye2013continuous} used a day-to-day learning model within a route choice setting to reveal the dynamic evolution process of traffic flow and credit price under a TCS. More recently, \citet{brands2020tradable} empirically tested the market of a TCS through a lab-in-the-field experiment where participants make virtual travel choices and real transactions in a tradable parking permits setting. The results showed that the designed market achieved credit prices within a desired range, and the number of bought and sold quantities kept close to each other, in accordance with a theoretical market equilibrium. Also recently, aiming towards a TCS closer to practice, \cite{trinity2020working} modelled the detailed and joint individual decision making (namely, the buying, selling and departure time choices) together with the regulator’s operations in a microscopic time-based simulation framework. Under this setting, the authors proposed a market model with desirable stability and equilibrium properties and compared against a time of day pricing control and the no-control case under different demand scenarios.
Ultimately, and in the absence of empirical evidence, these few existing contributions point to known state-of-the-art day-to-day learning frameworks for capturing some of the demand-supply interactions. Indeed,  different day-to-day models have been developed in recent years to evaluate other traffic control strategies, such as: \citet{cantarella2015day} studied the impact of different bus operating strategies and travelers' choices within a day-to-day dynamic process; \citet{guo2016day} investigated the properties of the day-to-day dynamic traffic flows in a general network under a congestion price; or \citet{yildirimoglu2020demand}  proposed departure time allocation optimization under both a day-to-day and a within-day framework. All these frameworks are often extremely helpful in ensuring desirable equilibrium properties in disaggregate modelling frameworks of demand-supply interactions, as it is for the TCS schemes at stake in this paper.

On the demand side, both detailed market interactions and day-to-day learning processes are still to be explored by researchers, whereas on the supply side, general and toy networks with link-based credit charging have prevailed in most studies, thus also limiting the design towards practice-ready TCS \citep{lessan2019credit}. In this paper we propose to examine the TCS under a single reservoir network using the Macroscopic Fundamental Diagram (MFD) \citep{daganzo2007urban,geroliminis2007macroscopic} and investigate corresponding properties. Early and recent MFD applications, such as \citep{arnott2013bathtub,fosgerau2013hypercongestion, Amirgholy2017215}, showcased the quantification of potential benefits in terms of congestion reduction under quantity control and pricing schemes by keeping the accumulation no greater than the flow-maximizing value. The individual attributes such as heterogeneous trip length was considered in \citep{arnott2013bathtub, fosgerau2015congestion,daganzo2015distance, lamotte2016morning}, where a reformulation of the computation for trip length is developed, henceforward referred to as \textit{trip-based MFD}. There are basically three advantages of utilizing the trip-based MFD model: 1) Compared to the traditional MFD model (or \textit{accumulation-based MFD model} \citep{leclercq2015macroscopic}), where the predicted outflow increases instantaneously when there is a sharp increase in the inflow, trip-based MFD model accounts for a reaction time to the sudden change in demand and computes the outflow only considering travelers who have completed their trips, providing more reliable results \citep{mariotte2017macroscopic}; 2) Trip-based MFD model is able to accommodate a more realistic heterogeneity of individual travelers in terms of trip length, desired arrival time, and schedule deviation penalties \citep{lamotte2018morning}; and 3) Trip-based MFD model allows for testing distance-based TCS schemes, bringing additional degrees of freedom to the scheme's design process towards fairness and efficiency \citep{daganzo2015distance}.

% Moreover, as pointed out by \citet{daganzo2015distance}, usage-based (length-specific) tolling is more efficient than length-agnostic tolling, it is promising to obtain further gains by applying the TCS to the trip-based MFD model.
Moreover, a proper design of credit charging scheme is required to make the TCS effective for demand management. The system optimal credit scheme are usually derived through a closed-form objective when dealing with a static traffic equilibrium model (e.g., \citet{yang2011managing, de2018congestion} for link-specific credit tolls, and \citet{xiao2013managing, bao2019regulating} for time-varying credit tolls). However, it is more complex to obtain the charging scheme that minimizes system cost for day-to-day dynamic choice (e.g., departure time) modelling problems. Though several studies have examined the problem of computing the link-specific congestion pricing toll to reach a desired equilibrium \citep{han2017discrete, liu2017doubly}, or to minimize the system cost \citep{tan2015dynamic}, it remains challenging to determine optimal time-varying charging schemes. \citep{Amirgholy2017215} presented both analytical and numerical approximations of the equilibrium solution under a time-varying road pricing scheme for a semi-quadratic \textit{accumulation-based MFD} and proved its accumulation maximization solution under limited heterogeneity conditions. \citet{liu2020bayesian} proposed to use a Gaussian (mixture) function to parameterize the time-varying road pricing scheme so as to facilitate the adoption of derivative-free optimization methods, such as evolution algorithms, pattern search, and Bayesian optimization. 

This paper incorporates the TCS and the trip-based MFD model into a day-to-day modelling framework to investigate the properties of the equilibrium solutions and the performance of the TCS at the network level. As proposed in \cite{trinity2020working}, the framework is then applied to the (departure-time choice) morning commute problem and used to evaluate welfare, network and traveler-specific performance changes. More specifically, the contributions of this paper are four-fold: (1) The day-to-day model is integrated with the trip-based MFD model under a TCS, wherein the credit charge rate is time-varying and length-specific, and the credit price is day-to-day dynamic. (2) The analytical properties of the proposed dynamic system are established, including the existence of the equilibrium solution and the uniqueness of the credit price at equilibrium state. (3) As the proposed day-to-day model is essentially a simulation with an expensive-to-evaluate objective function when dealing with a large number of travelers, Bayesian Optimization (BO) is adopted to optimize the credit charge scheme (which is referred to as toll profile) to maximize the social welfare. (4) A simulation framework is also presented to validate the analytical properties and evaluate the TCS performance against the case without TCS and the time-of-day pricing case.

The rest of this paper is organized as follows. Section \ref{S:2} presents the modeling assumptions, traffic flow model, day-to-day dynamic model, and credit price evolution model. Section \ref{S:3} discusses the properties of the dynamic system, i.e., the existence of the equilibrium solution and the uniqueness of the credit price at equilibrium state. Section \ref{S:4} introduces the BO approach to optimize the parameters associated with the toll profile function. Then the numerical simulation results are discussed in Section \ref{S:5}. Finally, conclusions are summarized in Section \ref{S:6}.

\section{Methodology}
\label{S:2}
\subsection{Characteristic analysis and modeling assumptions}

%Missing DEscription of the morning commute problem
In a morning commute problem, travelers are assumed to choose their departure times based on the travel time and schedule delay (the difference between the actual arrival time and desired arrival time) costs \citet{vickrey1969congestion}. In this paper, we follow the early philosophy in modelling the morning commute problem with area-based networks as a reservoir \citep{daganzo2007urban, gonzales2012morning} and investigate its properties under a TCS. Note that in this modelling framework, it is assumed that traffic congestion is spatially uniformly distributed within the network \citep{daganzo2007urban}.

%CLA
First, we present key definitions regarding the network, travelers' behavior and credit price, which are essential for the  modelling and analysis that follows.

\textit{Network}: A single-reservoir network \citep{daganzo2007urban,geroliminis2007macroscopic} is considered in this study, where all trips originate and end within this network. The idea is to describe the aggregate vehicular accumulation, the number of operating vehicles, at the "neighborhood" level with a well-defined relationship between the reservoir outflow and the aggregate accumulation. Furthermore, we will resort to the \textit{trip-based MFD} as in \cite{arnott2013bathtub, fosgerau2015congestion,daganzo2015distance, lamotte2016morning}, whose general properties are described in \ref{section:MFDmodel} and further investigated in \cite{mariotte2017macroscopic}. 

\textit{Traveler}: Each traveler (in a vehicle) is simulated as an individual agent. The initial travel plan of each agent consists of departure time, trip length and desired arrival time. In the day-to-day process, departure time can change within a fixed time-window, and a logit discrete choice model framework along the lines of \citep{ben1984dynamic} is used (see \ref{section:Choice}. The time-based attributes of the alternatives considered rely on the previous experienced attributes and the historical perception by each individual \citep{sheffi1984urban} as detailed in \ref{section:Learning}.

\textit{TCS}: The TCS explored in this paper follows the design proposed by \citep{bao2019regulating,brands2020tradable, trinity2020working}: 1) The regulator will give out the same amount of credits to every traveler in each day. 2) All travelers are assumed to trade directly with the regulator so that travelers who are short of credits can buy enough credits to pay the tariff, and travelers who have excess credits can sell them. 3) The tariff charging profile (i.e., the time-varying toll in credits) is designed by the regulator in advance and is invariant across the entire day-to-day process. Here, the optimal design of the toll profile is obtained using BO, further detailed in \ref{section:BO}.

\textit{Credit price}: As the credit is freely bought and sold in a market, the price is determined by credit demand-supply interactions \citep{yang2011managing,ye2013continuous}. Specifically, %\textcolor{red}{with the consideration of a fixed toll profile}, 
the supply of credits is predetermined by the regulator while the demand or consumption of credits is governed by the credit toll profile and the traffic flow generated from the executed travel plans of all the agents. The difference between the supply and demand determines the current credit price, which in turn influences the evaluations of agents (as detailed in Section \ref{section:Price}).

%CLA
Given the definitions above, the following assumptions are made:
% If we make these assumptions formal the way you did and not "en passant" then dont we need to discuss it? For example do we need to present justifications for it and not for alternative assumptions?

\textbf{Assumption 1.} The travel demand is assumed not to be excessively large that can trigger a gridlock. This assumption ensures that the gridlock is never reached and a non-zero flow stable equilibrium can be achieved.

\textbf{Assumption 2.} Each traveler has his/her own desired arrival time, trip length and schedule deviation early/late penalty coefficient \citep{lamotte2018morning}. This assumption ensures the heterogeneity within travelers, which can be accommodated by the trip-based MFD model as it computes the travelled distance of all travelers separately. Note that the size of the departure time choice set is the same for all travelers. A relaxation of this assumption will be considered in future work.

\textbf{Assumption 3.} The utility specification in the individual logit discrete choice model has a simple formulation with:  travel time, schedule deviation penalty, toll credit payment as attributes and the random utility term \citep{ben1984dynamic, ben1985discrete}.

\textbf{Assumption 4.} In the day-to-day process, when evaluating the utilities of the not-chosen alternatives (and thus not experienced), the attributes used are estimated by assuming that there are fictional travelers, who will not influence the accumulation of the network, entering the network at those alternative departure times (\citet{lamotte2015dynamic}).
%These traffic information should be available as the information and communications technologies have been greatly developed.
% RM: Now the travel time is estimated by using the experienced travel time of fictional travelers.%

\textbf{Assumption 5.} The buying and selling behaviors are not explicitly modelled in this study, in contrast with \citet{trinity2020working}. Thus, the heterogeneous preferences on the credit market or known behaviors existing in a tradable market such as explicit buys and selling and stocking are not considered. Here, for each day and after selecting a departure time, a traveler will have to pay a credit toll according to the toll profile function, her/his trip-length and the endowment. If the traveler is short of credits, she/he can only buy the credits needed for the payment directly from the regulator; otherwise, she/he will sell extra credits to the regulator at the market price. We are only assuming a credit balance of $0$ at the end of the day and we do not make an assumption on the explicit number of trades. This assumption is commonly adopted in previous research \citep{yang2011managing,wu2012design,brands2020tradable}. It is expected that complex and strategic buying and selling behaviors will affect credit price evolution \citep{dogterom2017tradable, trinity2020working} through both the credit value perception and the credit effective demand and supply. A common way to avoid the speculation is to set a specific validity period; then no one can benefit from credit stocking and banking behaviors \citep{de2020tradable}. In this study, the validity period is set as one day. Besides, the aforementioned end-of-day "traveler-to-regulator" trading has a main disadvantage that the budget neutrality of credits is not guaranteed \citep{brands2020tradable}. A possible solution could be setting a cap for the supply of credits from the regulator. Nevertheless, the proposed price adjustment mechanism in Section \ref{section:Price} can ensure market clearance of credits at the equilibrium, hence the budget neutrality will be guaranteed eventually, see Section \ref{S:5} for detailed results.

\textbf{Assumption 6.} To set up a time-varying credit charge scheme, the toll function is assumed to take the form of a (positive) Gaussian function, which is controlled by three parameters, mean, variance and amplitude. Without loss of generality, the method described below can be extended to a Gaussian mixture function to allow for asymmetric and more flexible toll profiles \citep{liu2020bayesian}. Moreover, alternative step toll and a triangular toll profiles are also tested under the proposed framework in Section \ref{S5_5_1}.

Under the above assumptions, the detailed model formulation that follows relies on the notation presented in Table \ref{notation}.

\begin{table}[H]
  \centering
  \caption{Modelling Notations}\label{notation}
  \vskip 0.2cm
  \begin{tabular}{ll}
\hline\noalign{\smallskip}
 Notations  & Definition \\
\noalign{\smallskip}\hline\noalign{\smallskip}
$i$ &  Index of a traveler\\
$d$ &  Index of a simulated day\\
$I_{i,d}$ & Number of credits distributed to traveler $i$ on day $d$\\
$D$ &  Total number of days\\
$N$ &  Total number of travelers\\
$L_i$ &  Trip length of traveler $i$\\
$w$ &  Scale factor of trip length\\
$\theta_i$ &  Value of time for traveler $i$\\
$SDE_i \backslash SDL_i$ &  Schedule deviation penalty for early$\backslash$late arrival for traveler $i$\\
$t_{i,d}^{dep}$ &  Departure time of traveler $i$ on day $d$\\
$\tau$ &  Time window size parameter\\
$\Delta t$ &  Step-size of departure time\\
$TW_i$ &  Departure time window for traveler $i$,\\ 
& $TW_i=\{t_{i,0}^{dep}-\tau\cdot\Delta t,t_{i,0}^{dep}-(\tau-1)\cdot\Delta t,\dots,t_{i,0}^{dep}+\tau\cdot\Delta t\}$\\
$n(t)$ &  Number of travelers in the network at time $t$\\
$V(n)$ &  Network's average travel speed with $n$ travelers\\
$T_i^*$ & Desired arrival time of traveler $i$\\
$T_{i,d}(t)$ & Experienced (or estimated) travel time for traveler $i$ on day $d$ \\
& departing at time $t$\\
%$tt_{i,d}(t)$ &   travel time for traveler $i$ on day $d$ departing at time %$t$\\
$c_{i,d}(t)$ & Experienced (or estimated) generalized cost for traveler $i$ on day $d$\\
& departing at time $t$\\
$C_{i,d}(t)$ &  Disutility associated with perceived generalized cost for traveler $i$\\
& on day $d$ departing at time $t$\\
$\delta_i$ & Binary variable, $\delta_i=1$ if $t+T_{i,d}(t)<T_i^*$, otherwise $\delta_i=0$\\
$\omega$ & Learning parameter for the generalized cost in the day-to-day process\\
$p_d$ &  Credit price on day $d$\\
$Toll(t)$ &  Number of credits charged per trip length unit at time $t$ \\
\hline
\end{tabular}
\end{table}

\subsection{Trip-based MFD model}\label{section:MFDmodel}
%\citet{lamotte2016morning} defined the "speed-MFD" that all travelers in the network move with mean speed $V(n)$, where the accumulation $n$ results from the departure and arrival event in the reservoir. 
As defined in \citep{lamotte2016morning,mariotte2017macroscopic}, the general principle of the \textit{trip-based MFD} is that the trip length of traveler $i$ is computed as the integration of the speed from the entering time $t_i^{dep}$ to the exiting time $t_i^{dep}+T_{i}(t_i^{dep})$, which is written as follows,
\begin{equation}
\label{trip_l}
L_i = \int_{t_i^{dep}}^{t_i^{dep}+T_{i}(t_i^{dep})} V(n(t))dt
\end{equation}

Without loss of generality, notation $t_i^{dep}$ is used instead of $t_{i,d}^{dep}$ in Equation (\ref{trip_l}). %This formulation follows \textbf{Assumption 1} that the desired arrival time, trip length and schedule deviation early/late penalty coefficient are considered as heterogeneous input parameters.
The assumption that the $V(n(t)))$ is the same for all travelers in the neighborhood-like network and only changes with an event (departure or arrival) often requires the use of event-based simulation for analysing network properties \citep{mariotte2017macroscopic,yildirimoglu2020demand}. In this paper we follow the simulation process below:
\begin{algorithm}[H]
\label{a21}
\caption{Event-based simulation of the trip-based MFD}
\emph{Step 1.} Initialization: Input $t_i^{dep}$, $T_i^*$, $L_i$, speed-MFD function $V(n)$ and number of travelers $N$; set $n=0$, event counter $j=0$, $t_j=0$; calculate the initially estimated arrival time for all travelers $\forall i , 1...N$ by $L_i/V(0)$.

\emph{Step 2.} Construct the event list by appending the departure and arrival in the order of time, which should have a length of $2N$.

\emph{Step 3.} Calculate the experienced travel time:

\textbf{While} Event list is not empty:

\quad set $j=j+1$

\quad set $t_j$ as the time of the next event
 
\quad let $L_i=L_i-V(n)\cdot (t_j-t_{j-1})$, $\forall i$
 
\quad \textbf{if} the next closest event is a traveler $i'$ departure:
 
\quad \quad $n=n+1$, update the credit account balance of traveler $i$
 
\quad \textbf{else}:
 
\quad \quad $n=n-1$, compute the experienced travel time of traveler $i'$, $T_{i'}(t_{i'}^{dep})$
 
\quad \textbf{end if}
 
\quad Remove this event from the event list
 
\quad Update the current average travelling speed $V(n)$
 
\quad Update the estimated arrival time for travelers currently in the network by equation \eqref{trip_l} considering a constant speed $V(n)$
 
\quad Sort the event list in the order of time
 
\textbf{End while}
\end{algorithm}

\subsection{Travel behavior model and and credit price evolution}
\subsubsection{Travelers' departure time choice}\label{section:Choice}
On day $d$, the utilities of the discrete departure time choice for traveler $i$ has the following formulation:
\begin{equation}\label{utility}
    U_{i,d}(t)= C_{i,d}(t)+ \epsilon_i
\end{equation}
where $\epsilon_i$ is an identically and independently distributed error term; and $C_{i,d}(t)$ is the systematic disutility associated with the perceived cost of traveler $i$ departing at time $t$ on day $d$, further detailed in Section \ref{lp}. The probability of choosing departure time $t$ can be calculated according to the logit model as follows:
\begin{equation}\label{prob}
    Pr_{i,d}(t)= \frac{\exp\big(\mu\cdot C_{i,d}(t)\big)}{\sum_{s\in TW_i}\exp\big(\mu\cdot C_{i,d}(s)\big)}
\end{equation}
where $\mu>0$ is the scale parameter, reflecting the variance of the unobserved portion of utility, with choices being random for scale parameter equal to zero and deterministic as its value approaches infinity \citep{sheffi1984urban, ben1985discrete}.

\subsubsection{Travelers' learning process}\label{section:Learning}
\label{lp}
Under a given TCS, the generalized cost for traveler $i$ on day $d$ departing at time $t$ consists of three parts: travel time cost, schedule deviation penalty and credit payment. More specifically, we formulate a money metric utility as:
\begin{equation}
\begin{aligned}
\label{exp_cost}
    c_{i,d}(t) =& -\theta_i\cdot \Big[T_{i,d}(t)+\delta_i\cdot SDE_i\cdot\big(T_i^*-t-T_{i,d}(t)\big)+\\
    &(1-\delta_i)\cdot SDL_i\cdot\big(t+T_{i,d}(t)-T_i^*\big)\Big]
    -p_d\cdot Toll(t)\cdot L_i\cdot w\\
    =& -\theta_i\cdot tc_{i,d}(t)-p_d\cdot Toll(t)\cdot L_i\cdot w
\end{aligned}
\end{equation}
% \begin{equation}\label{travel_time}
%     T_{i,d}(t)=\frac{tt_{i,d}(t_{i,d}^{dep})}{T_{i,d}^{ins}(t_{i,d}^{dep})}\cdot T_{i,d}^{ins}(t)
% \end{equation}
where $\theta_i$ is the value of time for traveler $i$, $T_{i,d}(t)$ is the travel time for traveler $i$ on day $d$ departing at time $t$, $SDE_i$ and $SDL_i$ are the schedule deviation penalty parameters for early and late arrival for traveler $i$, and $\delta_i$ is a binary variable that equals 1 if $i$ arrives early and 0 otherwise. For simplicity we also formulate $tc_{i,d}$ as the time related component of the utility, i.e.  the travel time cost and the schedule deviation penalty. Note that with \textbf{Algorithm 1}, only the travel time for the chosen departure time can be measured by a particular traveler $i$. In order to estimate the travel time for all other unchosen departure times in the choice set $TW_i$, fictional travelers are assumed to choose these departure time without influencing the accumulation of the network \citep{lamotte2015dynamic}. %The calculation of the travel time also follows equation (\ref{trip_l}) under \textbf{Assumption 3}. 
The credit payment of traveler $i$ is the product of credit price $p_d$, credit toll $Toll(t)$, trip length $L_i$ and the scale factor $w$, where $w$ is needed for scaling down the magnitude of trip length to avoid unrealistic large payment.

% \citet{yildirimoglu2020demand} proposed equation (\ref{travel_time}) as a correction for the estimation of travel time for all the departure time in the choice set $TW_i$. When $t=t_{i,d}^{dep}$, $T_{i,d}(t)$ equals to the experienced travel time $tt_{i,d}(t_{i,d}^{dep})$, i.e., $T_{i,d}(t)=T_{i,d}(t_{i,d}^{dep})=tt_{i,d}(t_{i,d}^{dep})$; otherwise, the estimated travel time $T_{i,d}(t)$ for traveler $i$ departing at other unchosen departure time in the choice set $TW_i$ is calculated based on the instantaneous travel time.
 
At the end of each day $d$, travelers update their perception of the generalized costs for day $d+1$, combining the initially perceived generalized costs on day $d$ with the experienced (chosen alternative) and estimated (unchosen alternatives) generalized costs on day $d$, as follows:
\begin{equation}\label{learning}
    C_{i,d+1}(t)=\omega\cdot C_{i,d}(t)+(1-\omega)\cdot c_{i,d}(t)
\end{equation}
where $0<\omega<1$ is the learning parameter. The learning rate implies the weight of previously measured cost \citep{horowitz1984stability}. A larger learning rate means a higher influence of historical perceived costs. In this study, all travelers are assumed to have the identical learning rate $\omega$.

\subsubsection{Credit price evolution}\label{section:Price}
Price in a free market is typically determined by the relationship between supply and demand, i.e. the market price should increase when demand exceeds supply, and reduce when supply exceeds demand. In this study, it is assumed that the credit price on day $d+1$ is based on the previous day's credit price $p_d$ and the observed excess credit consumption $Z_d$, defined as the difference between the credit consumption and the total endowment of credits in a given day. The endowment $I_{i,d}$ is assumed to be the same for all travelers and constant across days. Thus,
\begin{align}\label{price}
% &EMA_d=\omega\cdot EMA_{d-1}+(1-\omega)\cdot Z_d\\
&Z_d=\sum_i Toll(t_{i,d}^{dep})\cdot L_i\cdot w- I_{i,d}\cdot N\\
&p_{d+1}=p_{d}+Q(p_{d}, Z_d)
\end{align}
where the change of price is represented by function $Q$.

Function $Q(p,Z)$ needs to satisfy the following assumptions to guarantee a non-negative price $p$: 1) $\forall p>0$, $Q(p,Z)$ is strictly increasing with $Z\in\mathbb{R}$; 2) If $p=0$, $Q(0,Z)$ is strictly increasing with $Z\geq 0$; 3) $Q(p,0)=0,~\forall p\geq 0$; $Q(0,Z)=0,~\forall Z\leq 0$. Then, we have
\begin{equation}\label{funQ}
    Q(p,Z)=0\Longleftrightarrow p\cdot Z=0,\quad p\geq 0,~Z\leq 0
\end{equation}

\section{Solution analysis}
\label{S:3}
This section builds on the work of \citet{ye2013continuous}, who showed that the equilibrium of market price and flow dynamics under a TCS in a route choice context is unique and stable.

%Furthermore, all analytical properties here analyzed rely on the assumption of continuous formulation. Then in Section \ref{S:5} these properties are validated by numerical simulations.

\subsection{Solutions to the day-to-day model}
\label{S31}
When the proposed day-to-day model reaches equilibrium, the total number of consumed credits should not exceed the total number of endowed credits, which can be written as follows:
\begin{equation}\label{feasibility}
    \lim_{d\rightarrow D} \sum_i Toll(t_{i,d}^{dep})\cdot L_i\cdot w\leq I\cdot N
\end{equation}
where $I_{i,d}$ is simplified as $I$ as it is identical among travelers and constant across days.

As the credit toll profile within the peak hour period is predetermined and constant across days, it allows for the calculation of the theoretical minimum consumption of credits for all travelers. In this case, the dynamic system will be self-adaptive to the credit endowment, as long as the feasibility condition that credit endowment is not smaller than the minimum demand is satisfied.

\subsection{Existence of the equilibrium point and uniqueness of the price}
\label{property}
For the proposed dynamic system, the equilibrium point $C_{i,*}(t)$ means the vector of the perceived generalized cost of all travelers $\bm{C}_{d}$ is equal to the vector of the experienced generalized cost of all travelers $\bm{c}_{d}$, i.e., when the system comes to $C_{i,*}(t)$, it will remain at $C_{i,*}(t)$ for all future times. By the definition, for all traveler $i\in\{1,2,\dots,N\}$, the equilibrium condition for this system is:
\begin{equation}\label{equilibrium1}
\left\{
\begin{aligned}
     &C_{i,*}(t)=\omega C_{i,*}(t)+(1-\omega)\cdot [\theta_i\cdot tc_{i,*}(t)+p_*\cdot Toll(t)\cdot L_i\cdot w]\\
     &p_{*}=p_{*}+Q(p_{*}, Z_*)\\
\end{aligned}
\right.
\end{equation}
which is equivalent to 
\begin{equation}\label{equilibrium2}
\left\{
\begin{aligned}
     &C_{i,*}(t)=\theta_i\cdot tc_{i,*}(t)+p_*\cdot Toll(t)\cdot L_i\cdot w\\
     &Q(p_{*}, Z_*)=0\\
\end{aligned}
\right.
\end{equation}

We can then introduce the following two theorems for the existence of the equilibrium point and the uniqueness of the price.

\begin{theorem}\label{thm1}
If \textbf{Assumption 1} holds, and if $I\cdot N>I_{\min}\cdot N\triangleq \lim_{p\rightarrow \infty} \sum_iToll(t_i)\cdot L_i\cdot w$, then there exists at least one equilibrium point $(C_{i,*},p_*)$ of the proposed dynamic system.
\end{theorem}

\begin{proof1}
The proof is detailed in \ref{P_thm1}. It is directly inspired from a proof of \citet{ye2013continuous}. \qed
\end{proof1}

\begin{theorem}\label{thm2}
Assume the minimum credit endowment condition is satisfied, if the total credit demand, $\sum_i Toll(t_{i,*}^{dep})\cdot L_i\cdot w$, is strictly decreasing with credit price, i.e.,
\begin{equation}\label{p_condition}
    (p_1-p_2)\Big(\sum_i Toll(t_{i,*}^{dep_1})\cdot L_i\cdot w-\sum_i Toll(t_{i,*}^{dep_2})\cdot L_i\cdot w\Big)<0
\end{equation}
where $t_{i,*}^{dep_1}$ and $t_{i,*}^{dep_2}$ is the departure time of traveler $i$ under price $p_1$ and $p_2$, respectively, then the equilibrium price is unique.
\end{theorem}

\begin{proof2}
The proof is detailed in \ref{P_thm2}. \qed
\end{proof2}

\begin{remark}\label{rm1}
Notes on the conditions for inequality (\ref{p_condition}) to hold:

1) There is at least two departure times with different credit charges. This condition holds as the credit toll is time-varying.

2) According to equation (\ref{prob}), for each traveler $i$, the probability of choosing departure time $t$, $Pr_{i}(t)=Pr_{i}(C_{i}(t))$ takes the logit form, thus the following conditions are satisfied:
\begin{tabenum}[i)]
\qquad\qquad \tabenumitem $ \displaystyle \frac{\partial Pr_{i}(C_{i}(t))}{\partial C_{i}(t)} > 0 \quad \forall t\in TW_i,~i\in \{1,2,\dots,N\}$\\
\qquad\qquad \tabenumitem $ \displaystyle \frac{\partial Pr_{i}(C_{i}(t_1))}{\partial C_{i}(t_2)} < 0 \quad \forall t_1,~t_2\in TW_i,~i\in \{1,2,\dots,N\}, t_1\neq t_2 $\\
\qquad\qquad \tabenumitem $ \displaystyle \frac{\partial Pr_{i}(C_{i}(t_1))}{\partial C_{j}(t_2)} = 0 \quad \forall t_1\in TW_i,~t_2\in TW_j,~i\neq j\in \{1,2,\dots,N\} $\\
\qquad\qquad \tabenumitem $ \displaystyle \frac{\partial Pr_{i}(C_{i}(t_1))}{\partial C_{i}(t_2)} = \frac{\partial Pr_{i}(C_{i}(t_2))}{\partial C_{i}(t_1)} \quad \forall t_1,~t_2\in TW_i,~i\in \{1,2,\dots,N\} $\\
\qquad\qquad \tabenumitem $ \displaystyle \sum_{t'\in TW_i}Pr_{i}(C_{i}(t'))=1, \quad \forall i\in \{1,2,\dots,N\} $
\end{tabenum}

Condition i) shows the partial derivative of $Pr_{i}(C_{i}(t))$ with regard to $C_{i}(t)$ is positive, which means that the probability of choosing departure time $t$ becomes higher if the $C_{i}(t)$ increases (perceived cost decreases). Condition ii) means that for traveler $i$ the probability of choosing departure time $t_1$ becomes lower if for another time $t_2$, $C_{i}(t_2)$ increases (perceived cost of $t_2$ decreases). Condition iii) represents the fact that the decision-making process of traveler $i$ is not affected by other travelers. In addition, it is trivial to derive that $\frac{\partial Pr_{i}(C_{i}(t_1))}{\partial C_{i}(t_2)}= Pr_{i}(C_{i}(t_1))\cdot Pr_{i}(C_{i}(t_2)), t_1\neq t_2$, then condition iv) holds. Finally, the probabilities of all alternatives should sum up to 1, as shown in condition v).

3) Let $\bm{J_{Pr}}$ be the Jacobian matrix of $\bm{Pr}(\cdot)$ with regard to $\bm{C_i}(\bm{t})=(C_{i}(t'),~t'\in TW_i)$. It can be proved that $\bm{J_{Pr}}$ is negative semidefinite as condition 2) is satisfied and the systematic utility is a non-increasing function of the corresponding perceived cost (see details in \citep{cantarella1995dynamic,ye2013continuous}). According to the property of semidefinite matrix, we have $\bm{Toll}\cdot \bm{J_{Pr}}\cdot \bm{Toll}^T\leq 0$, where $\bm{Toll}=(L_i\cdot w\cdot Toll(t'),~t'\in TW_i)$. The equality holds if and only if condition 1) is not satisfied. Thus the following inequality holds:
\begin{equation}
\begin{aligned}
\label{p_condition2}
    & (p_1-p_2)\Big(\bm{Pr}(\bm{C_i}(\bm{t}|p_1))\bm{Toll}^T-\bm{Pr}(\bm{C_i}(\bm{t}|p_2))\bm{Toll}^T\Big)\\
    =& (p_1-p_2)\Big(\sum_{t\in TW_i} Pr_i(C_i(t|p_1))\cdot Toll(t)\cdot L_i\cdot w-\\
    &\sum_{t\in TW_i} Pr_i(C_i(t|p_2))\cdot Toll(t)\cdot L_i\cdot w\Big)\\
    < & 0
\end{aligned}
\end{equation}
In inequality (\ref{p_condition2}), $\sum_{t\in TW_i} Pr_i(C_i(t|p_1))\cdot Toll(t)\cdot L_i\cdot w$ is the expected value of the credit toll paid by traveler $i$ with credit price $p_1$. According to the weak law of large numbers, we can substitute the summation term by the credit toll paid by traveler $i$ departing at the equilibrium departure time $t_{i,*}^{dep_1}$ with credit price $p_1$. Then inequality (\ref{p_condition}) holds. \qed
\end{remark}

Inequality (\ref{p_condition}) can also be interpreted that under the logit form choice probability function, when the credit price goes up, the departure time that associates with a relative high toll is less likely to be chosen and the credit consumption is supposed to decrease.

Then, it is reasonable to propose the following hypothesis:
\begin{hyp}\label{hyp1}
$p_*$ is decreasing with $I_i$ in interval $(I_{\min},~I_{UE}]$, where $I_{\min}$ is defined in Section \ref{S31}, and $I_{UE}$ is the average credit consumption for equilibrium pattern without TCS. Besides, $p_*=0$ when $I_i\geq I_{UE}$. 
\end{hyp}

A similar hypothesis is analytically proved in \citep{ye2013continuous}, where the TCS is applied to a link-based network congestion model, instead, we will validate this hypothesis by simulations in Section \ref{S:5}.
% \subsection{Stability of the equilibrium point}

\section{Simulation-based optimization framework}
\label{S:4}
\subsection{Framework}
\label{S_4_1}
In this study, the social welfare per capita $W$ at the equilibrium state is adopted to measure the performance of scenarios with and without TCS. First, in the no toll case (or NTE, the scenario without TCS), the social welfare per capita $W_{NTE}$ is the travel consumer surplus (CS) for each traveler, i.e., the sum of observed travel utilities $U_{i,d}(t_{i,d}^{dep})$, including travel time cost, schedule deviation penalty and random utility, which is written as
\begin{equation}
\begin{aligned}
\label{sw1}
    W_{NTE} =& CS=\frac{1}{N}\sum_{i=1}^N U_{i,d}(t_{i,d}^{dep})\\
    =&\frac{1}{N}\sum_{i=1}^N\Big[-\theta_i\cdot \Big(T_{i,d}(t_{i,d}^{dep})+\delta_i\cdot SDE_i\cdot\big(T_i^*-t_{i,d}^{dep}-tt_{i,d}(t_{i,d}^{dep})\big)+\\
    &(1-\delta_i)\cdot\big(t_{i,d}^{dep}+tt_{i,d}(t_{i,d}^{dep})-T_i^*\big)\Big)+\epsilon_i(t_{i,d}^{dep})\Big]\\
    =& \frac{1}{N}\sum_{i=1}^N\Big[-\theta_i\cdot tc_{i,d}(t_{i,d}^{dep})+\epsilon_i(t_{i,d}^{dep})\Big]\\
    =& \frac{1}{N}\sum_{i=1}^N U'_{i,d}(t_{i,d}^{dep})
\end{aligned}
\end{equation}
where $d$ is the day when the system reaches the equilibrium. 

For the TCS scenario, the social welfare is computed as the sum of the CS, the travelers' revenue (TR) from unused credits, the regulator revenue (RR) from credit toll collection, the regulator cost (RC) from buying unused credits and the value of the travel endowment (TE). Let $\phi_{i,d}(t_{i,d}^{dep})$ represent the number of credits sold by traveler $i$ on day $d$, $\psi_{i,d}(t_{i,d}^{dep})$ denote the number of credits bought by traveler $i$ on day $d$, and $Toll_{i,d}^e(t_{i,d}^{dep})$ denote the number of credits paid by traveler $i$ on day $d$ from the endowment. Then we can compute the above social welfare components: $TR=\sum_{i=1}^N\phi_{i,d}(t_{i,d}^{dep})\cdot p_d$,  $RR=\sum_{i=1}^N\psi_{i,d}(t_{i,d}^{dep})\cdot p_d$, $RC=\sum_{i=1}^N\phi_{i,d}(t_{i,d}^{dep})\cdot p_d$, $TE=\sum_{i=1}^N Toll_{i,d}^e(t_{i,d}^{dep})\cdot p_d$. In addition, the paid credits by traveler $i$ can be considered as the sum of two components, the credits from endowment and credits bought from the regulator, i.e., $Toll(t)\cdot L_i\cdot w=\psi_{i,d}(t_{i,d}^{dep})+Toll_{i,d}^e(t_{i,d}^{dep})$. Based on the considerations above, the social welfare per capita $W_{TCS}$ is calculated as follows,
\begin{equation}
\begin{aligned}
\label{sw2}
    W_{TCS} =& CS+TR+RR-RC+TE\\
    =& \frac{1}{N}\sum_{i=1}^N\Big[-\theta_i\cdot tc_{i,d}(t_{i,d}^{dep})-p_d\cdot Toll(t_{i,d}^{dep})\cdot L_i\cdot w+\epsilon_i(t_{i,d}^{dep})\Big]+\\
    & \frac{1}{N}\Big[\sum_{i=1}^N\phi_{i,d}(t_{i,d}^{dep})\cdot p_d+ \sum_{i=1}^N\psi_{i,d}(t_{i,d}^{dep})\cdot p_d-\sum_{i=1}^N\phi_{i,d}(t_{i,d}^{dep})\cdot p_d+\\
    & \sum_{i=1}^N Toll_{i,d}^e(t_{i,d}^{dep})\cdot p_d\Big]\\
    =& \frac{1}{N}\sum_{i=1}^N\Big[-\theta_i\cdot tc_{i,d}(t_{i,d}^{dep})+\epsilon_i(t_{i,d}^{dep})\Big]\\
    =& \frac{1}{N}\sum_{i=1}^N U'_{i,d}(t_{i,d}^{dep})
\end{aligned}
\end{equation}
Note that the TCS welfare measure is equivalent to the NTE case and equals the combination of travel time cost and schedule deviation penalty plus the random utility component. %This ensures the validity of the comparisons between the no toll case and TCS case.

Our simulation captures the system model presented in Section \ref{S:2} including all detailed traveler and regulator states in the single-reservoir network under the designed market conditions\footnote{\label{note1}The open source code for the simulation is available at \url{https://github.com/RM-Liu/MFD_TCS}.}.
A simulation-based optimization was then conducted to find the optimal credit toll scheme which leads to the maximum social welfare $W_{TCS}$. Note that gradient-based algorithms will be computationally expensive for our optimization problem because they involve numerical computation of derivatives for an objective that has no closed form and is the outcome of a complex simulation. Furthermore, the simulation framework can be a time-consuming process under large number of scenarios. For this reason, the Bayesian optimization (BO) approach proposed in \citep{liu2020bayesian} for time-varying road pricing mechanisms is here adopted as it proved to approximate the simulation-based objective function using few evaluations. Furthermore, a Gaussian function is adopted to parameterize the credit toll profile and decrease the number of decision variables.

\subsection{Bayesian optimization}\label{section:BO}
A BO framework essentially consists of two main steps \citep{frazier2018tutorial}: Update a Bayesian statistical model that approximates a complex map from the input (i.e., the toll profile parameters: mean, variance and amplitude) to the output (i.e., the the social welfare $W_{TCS}$); determine the next input by optimizing an acquisition function. These two steps are discussed in details in Section \ref{GP} and \ref{Acq}, respectively.

\subsubsection{Gaussian Process}
\label{GP}
Invariably, a Gaussian Process (GP), which in our case assumes the objective function values and input points are joint distributed, is adopted in BO. Then the GP is fully specified by its mean function $\mu(\bm{x})$ and covariance function $k(\bm{x},\bm{x'})$ as follows,
\begin{equation}\label{GP_dist}
    W(\bm{x})~\sim~\mathcal{GP}\Big(\mu(\bm{x}),~k(\bm{x},\bm{x'})\Big)
\end{equation}
where $\bm{x}$ represents the input point, which is the credit toll scheme parameters.

To simplify the training and prediction of a GP model, the mean function $\mu(\cdot)=0$ is usually adopted \citep{williams2006gaussian}. Let us assume we have $m$ evaluated objective values according to a space-filling design of experiment as $\mathcal{D}_m=\{\bm{x}_{1:m},W_{1:m}\}$, where $\bm{x}_{1:m}=[\bm{x}_1,\bm{x}_2,\dots,\bm{x}_m]^T$ are the input points and $W_{1:m}=[W_1,W_2,\dots,W_m]^T$ are the corresponding objective values. Then the joint distribution of $W_{1:m}$ and the inference at a candidate input point $\bm{x}_{m+1}$, $W_{m+1}$, is as follows:
\begin{equation}\label{joint_dist}
\left[\begin{array}{c}
W_{1:m}\\
W_{m+1}
\end{array} 
\right]~\sim~\mathcal{N}\Bigg(\bm{0},~\left[\begin{array}{cc}
\bm{K} & \bm{k}\\
\bm{k}^T & k(x_{m+1},x_{m+1})
\end{array} 
\right]\Bigg)
\end{equation}
where $\bm{k}=[k(\bm{x}_{m+1},\bm{x}_1),k(\bm{x}_{m+1},\bm{x}_2),\dots,k(\bm{x}_{m+1},\bm{x}_m)]^T$, and $\bm{K}$ is the covariance matrix with entries $\bm{K}_{i,j}=k(\bm{x}_{i},\bm{x}_j)$ for $i,~j\in \{1,2,\dots,m\}$.

The posterior distribution of $W_{m+1}$ can be computed using Bayes' theorem,
\begin{equation}\label{posterior}
    W_{m+1}|W_{1:m}~\sim~\mathcal{N}\Big(\mu(\bm{x}_{m+1}),~\sigma^2(\bm{x}_{m+1})\Big)
\end{equation}
where $\mu(\bm{x}_{m+1})=\bm{k}^T\bm{K}^{-1}W_{1:m}$ and $\sigma^2(\bm{x}_{m+1})=k(x_{m+1},x_{m+1})-\bm{k}^T\bm{K}^{-1}\bm{k}$.

The covariance function is used to measure the correlation between two input points, that points closer in the input domain have a larger positive correlation and have more similar objective values \citep{frazier2018tutorial}. In this study, we choose the commonly used Matern kernel \citep{matern2013spatial},
\begin{equation}\label{kernel}
    k(\bm{x},\bm{x'})=\frac{2^{1-\nu}}{\Gamma(\nu)}(\sqrt{2\nu}\parallel\bm{x}-\bm{x'}\parallel)^{\nu}H_{\nu}(\sqrt{2\nu}\parallel\bm{x}-\bm{x'}\parallel)
\end{equation}
where $\Gamma(\cdot)$ is the Gamma function and $H_{\nu}$ is the modified Bessel function. In the following numerical examples, $\nu=5/2$ is used for the optimization process.

\subsubsection{Acquisition function}
\label{Acq}
After updating the posterior distribution over $W$, based on which the acquisition function measures the value of the candidate input points by using the inferred corresponding objective value and variance of the prediction.

Among the most popular acquisition functions, the upper confidence bound (UCB) \citep{srinivas2009gaussian} is used in this study, which is written as follows:
\begin{equation}\label{ucb}
    \alpha_{UCB}(\bm{x};\beta)=\mu(\bm{x})+\beta\sigma(\bm{x})
\end{equation}
where $\beta$ is a hyperparameter which controls balance between exploration and exploitation, that a larger $\beta$ will lead to more exploration. Then the next point to be evaluated can be determined by maximizing the UCB function (\ref{ucb}),
\begin{equation}\label{ucb_suggest}
    \bm{x}_{m+1}=\arg\max \alpha_{UCB}(\bm{x}_{m+1})
\end{equation}

\subsection{Experiment Design}
When conducting the BO from scratch, the initial input point is usually randomly generated from the input space, which may influence the solution quality and optimization efficiency. Then a space-filling design of the experiment could be useful for providing a good initial input point.

In this study, one of the most popular sampling method, Latin Hypercube Sampling (LHS) \citep{mckay2000comparison}, is used to generate the initial sample points. LHS stratifies each variable of $\bm{x}$ into $m$ equal intervals, and draws sample points from each sub-intervals uniformly. Compared to the Monte Carlo method, LHS has the advantage that the sampled points are independent without overlap, which provides a representative of the real variability.

\section{Numerical Experiments}
\label{S:5}
After presenting the simulation setting in the following subsection, we present here the results of (1) the day-to-day model properties and convergence (Section \ref{S5_2}) under a given credit toll profile; (2) the comparison between the optimized TCS against the NTE (\ref{S5_3}); (3) the comparison between the optimized TCS and time-of-day pricing (or CP) (\ref{S5_4}); and finally (4) the comparison among alternative types of credit toll profiles and credit schemes (\ref{S5_5}).

\subsection{Experiment settings}\label{S5_1}
The settings are presented in Table \ref{settings}.

\begin{table}[H]
  \centering
  \caption{Numerical settings}\label{settings}
  \vskip 0.2cm
  \begin{tabular}{ll}
\hline\noalign{\smallskip}
Parameters  & Specification \\
\noalign{\smallskip}\hline\noalign{\smallskip}
Credit endowment & $I_{i,d}=5$ [credit]\\
Demand &  $N_1=3000$ [traveler], $N_2=3700$ [traveler]\\
Initial departure time & $t_{i,0}^{dep}=\mathcal{N}(80,18)$, $t_{i,0}^{dep} \in [20,150]$\\
Trip length &  $L_i=4600+\mathcal{N}(0,(0.2*4600)^2)$ [m], $L_i>0$\\
Scale factor of trip length & $w=2\times 10^{-4}$ \\
Value of time & $\theta_i=1.1$ [DKK/min] \citep{fosgerau2007danish}\\
Schedule deviation penalty & $\begin{bmatrix} SDE_i \\ SDL_i \end{bmatrix}=\begin{bmatrix} 0.5\\ 4 \end{bmatrix}+\mathcal{N}\Bigg(\begin{bmatrix} 0.05^2 & 0.1^2\\ 0.1^2 & 0.4^2 \end{bmatrix}\Bigg)$\\
& $SDE_i \in [0.3,0.7]$, $SDL_i \in [2.5,5.5]$\\
Time window parameter & $\tau=30$ \\
Departure time interval & $\Delta t=1$ [min]\\
Network capacity & $n_{jam}=4500$ [vehicle] \\
Free flow speed & $v_f=9.78$ [m/s] \\
Speed function & $V(n)=v_f(1-\frac{n}{n_{jam}})^2$ [m/s]\\
Learning parameter &  $\omega=0.7$\\
Function $Q(p,Z)$ & $Q(p,Z)=kZ$, if $p>0$, $Q(p,Z)=\max\{0,kZ\}$\\
& if $p=0$, where $k=2\times 10^{-4}$\\
Toll profile function & $Toll(t|A, \xi, \sigma)=A\times e^{\frac{-(t-\xi)^2}{2\sigma^2}}$\\
\hline
\end{tabular}
\end{table}

This experiment considers a single-reservoir network with a capacity of 4500 travelers, with the speed function adopted from \citet{lamotte2018morning} and other parameters used in \citet{yildirimoglu2020demand}. The MFD is also characterized by the critical value of the accumulation, that can be computed according to the adopted speed function, i.e., $n_{cr}=1500$ travelers. Two demand scenarios, moderate congestion ($N_1=3700$ travelers) and high congestion ($N_2=4500$ travelers), are considered, where the critical value of the accumulation $n_{cr}$ is exceeded in both scenarios. The profiles for the accumulation of the two scenarios are shown in Section \ref{S5_2}. The initial departure time $t_{i,0}^{dep}$ is generated from a truncated Gaussian distribution. The desired arrival time $T_i^*$ is then computed as $t_{i,0}^{dep}+L_i/v_f$ for all travelers, which is also normally distributed.

Additionally, heterogeneous travelers are captured by drawing their trip lengths and schedule deviation penalties from three truncated Gaussian distributions, respectively. In both demand scenarios, the same distributions are used while all other parameters are constant (see Table \ref{settings}).

\subsection{Day-to-day evolution process}
\label{S5_2}
\subsubsection{Day-to-day process without TCS}
\label{without_TCS}
In this subsection, we first focus on the equilibrium of the day-to-day dynamics without TCS in both the moderate congestion scenario and high congestion scenario. When the day-to-day evolution reaches an equilibrium, the vector of the perceived generalized cost of all travelers $\bm{C}_{d}$ is equal to the vector of the experienced generalized cost of all travelers $\bm{c}_{d}$. Then the inconsistency between $\bm{C}_{d}$ and $\bm{c}_{d}$ is used as the indicator of equilibrium. Specifically, the L1 norm of the difference between them divided by the number of travelers $N$, i.e., $\parallel \bm{C}_{d}-\bm{c}_{d}\parallel_1/N$, is computed to represent the inconsistency. In addition, the gap of this generalized value is computed by $\parallel \bm{C}_{d}-\bm{c}_{d}\parallel_1/\parallel \bm{C}_{d}\parallel_1\times 100\%$.

Figure \ref{NTE3700}(a) presents the convergence of the perceived generalized cost of all travelers $\bm{C}_{d}$. It is found that the inconsistency becomes stable and close to 0 after 20 days, with a gap around 0.03\%, which implies that the day-to-day evolution reaches an equilibrium state. Figure \ref{NTE3700}(b) illustrates the evolution process of the average travel consumer surplus (\textit{CS}) per capita across days, and Figure \ref{NTE3700}(c) shows the evolution process of the social welfare per capita. These two plots show a same curve since the social welfare equals the consumer surplus in the no toll case. In addition, Figure \ref{NTE3700}(d) plots the evolution process of the average travel time cost across days. Besides, Figure \ref{NTE3700}(e) demonstrates the departure rates for every 5-minute interval on different days, and Figure \ref{NTE3700}(f) depicts the states of accumulation on different days, where the accumulation on day 20 overlaps with that of day 49. This also testifies that the equilibrium state is reached. In these two plots, the curves of day 0 represent the initial generated state specified in Section \ref{S5_1}, thus for the base case, it is not necessary to compare the equilibrium state with the initial state. Note that, as travelers do not have perceived cost and use predetermined departure time on day 0, the inconsistency, consumer surplus, social welfare and travel time cost are computed from day 1.

\begin{figure}[H]
    \centering
    \includegraphics[width=1\textwidth]{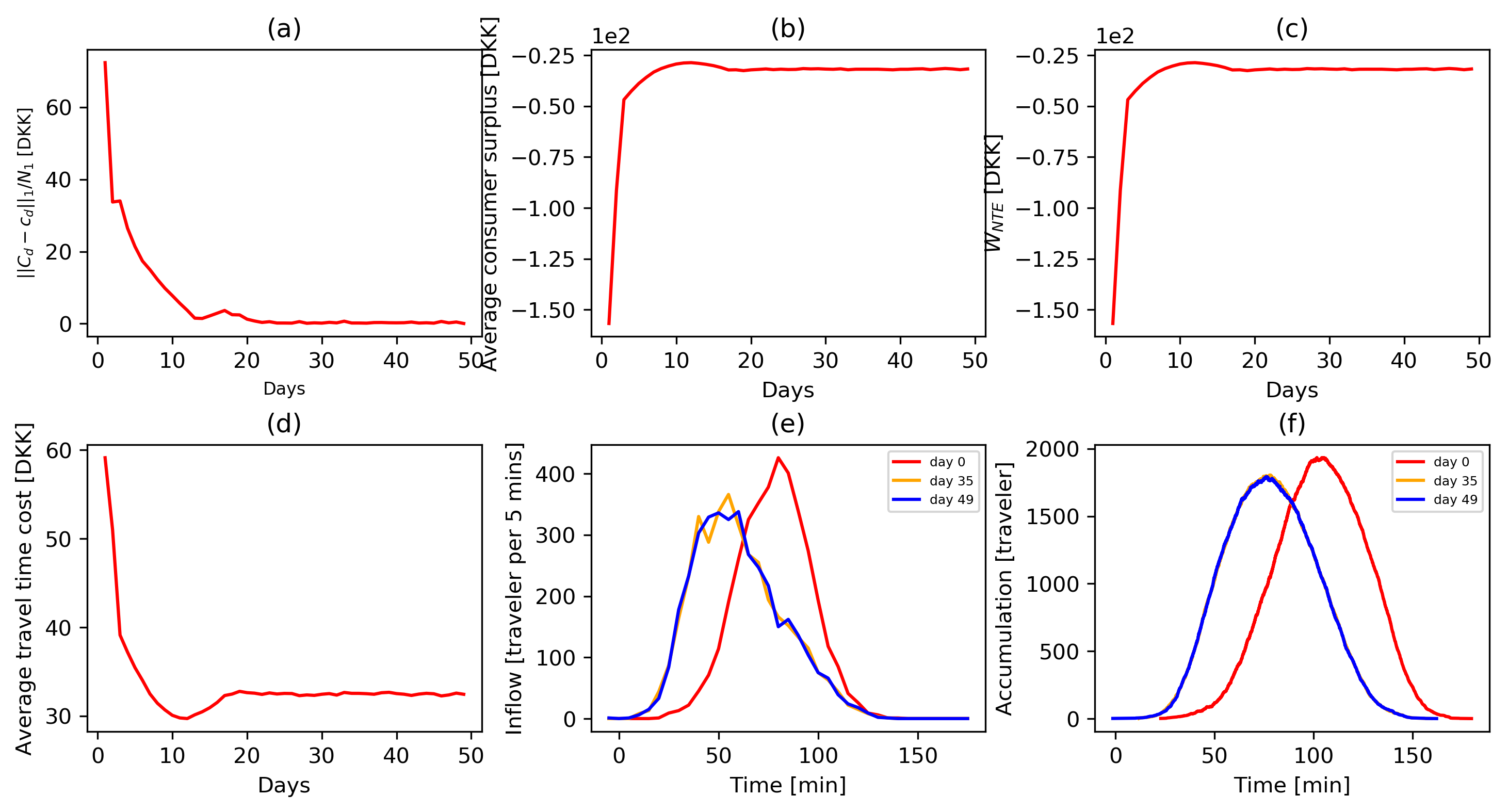}
    \caption{The evolution process of moderate congestion scenario without TCS}
    \label{NTE3700}
\end{figure}

Moreover, we further investigate convergence properties for the high congestion scenario. Figure \ref{NTE4500}(a) demonstrates that with the same learning parameter $\omega=0.7$, the perceived generalized cost is also able to reach a stable state. Figure \ref{NTE4500}(b-d) present the evolution of the average consumer surplus, social welfare per capita and average travel time cost, respectively. By comparison, we observe a lower consumer surplus, lower social welfare and higher travel time cost at the equilibrium state due to severer congestion. In addition, Figure \ref{NTE4500}(e) shows the departure rates, where there is a higher peak compared to the moderate congestion scenario, and Figure\ref{NTE4500}(f) shows the changes of accumulation on different days, where the peak accumulation at the equilibrium state exceeds the one in Figure\ref{NTE3700}(f) and is much higher than the critical value $n_{cr}$.

\begin{figure}[H]
    \centering
    \includegraphics[width=1\textwidth]{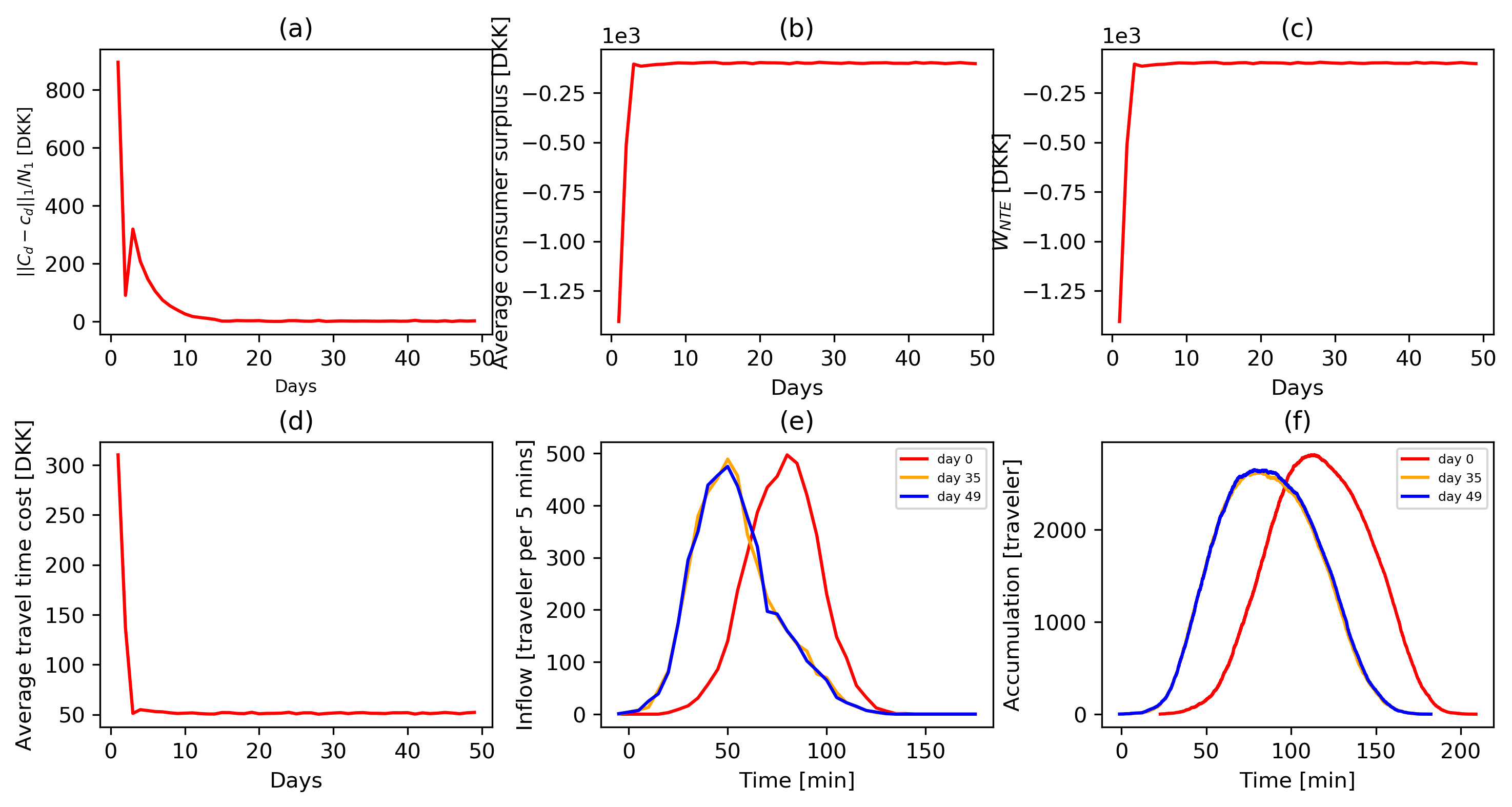}
    \caption{The evolution process of high congestion scenario without TCS}
    \label{NTE4500}
\end{figure}

\subsubsection{Day-to-day evolution with TCS}
\label{with_TCS}
In this subsection, we will present the convergence and the equilibrium properties mentioned in Section \ref{property} of the day-to-day dynamics with TCS for the moderate congestion scenario and high congestion scenario. It is worth to note that the equilibrium states of the base cases are used as the starting states (i.e., day 0) of the TCS cases for both demand scenarios. The parameters of the toll profile function are set as: $A=11$, $\xi=18$ and $\sigma=80$.

According to Figure \ref{TCS4500}(a)-(d), it is found that the day-to-day evolution under a given TCS also converges to an acceptable degree, where the system reaches the equilibrium after 20 days, with the gap equals to 0.4\%. Note that the social welfare $W_{TCS}$ in Figure \ref{TCS4500}(c) is smaller than that in the no toll case. This is because the credit toll scheme is not optimal but arbitrarily given, which is shown as a gray dashed line in Figure \ref{TCS4500}(f). The results of Figure \ref{TCS4500}(e) and (f) also support this observation that the peak departure rate and accumulation are not reduced compared to day 0. In addition, Figure \ref{TCS4500}(g) displays the evolution of the credit price, which goes up from 0 [DKK] to 4.6 [DKK] first and then decreases to the equilibrium price, which is 3.1 [DKK]. This process is consistent with the evolution of the credit transactions shown in Figure \ref{TCS4500}(h). At first, the number of bought credits (i.e., the credit demand) is much higher than the number of sold credits (i.e., the credit supply), implying that the market is short of credits and travelers need to buy extra credits from the regulator. Thus, the credit price increases. After perceiving a high travel cost due to the relatively expensive credit payment, travelers adjust their departure time to avoid being charged that much, leading to a smaller credit demand and consequently lower credit price. Finally, at the equilibrium state, the credit supply nearly equals to the credit demand and the correspondingly credit prices becomes stable. Figure \ref{TCS4500}(i) illustrates the evolution process of the average credit payment (which is referred as toll payment later), which is the value of the used credits. Note that this toll payment is exactly the difference between the consumer surplus and the social welfare, as shown in equation \eqref{sw2}.

\begin{figure}[H]
    \centering
    \includegraphics[width=1\textwidth]{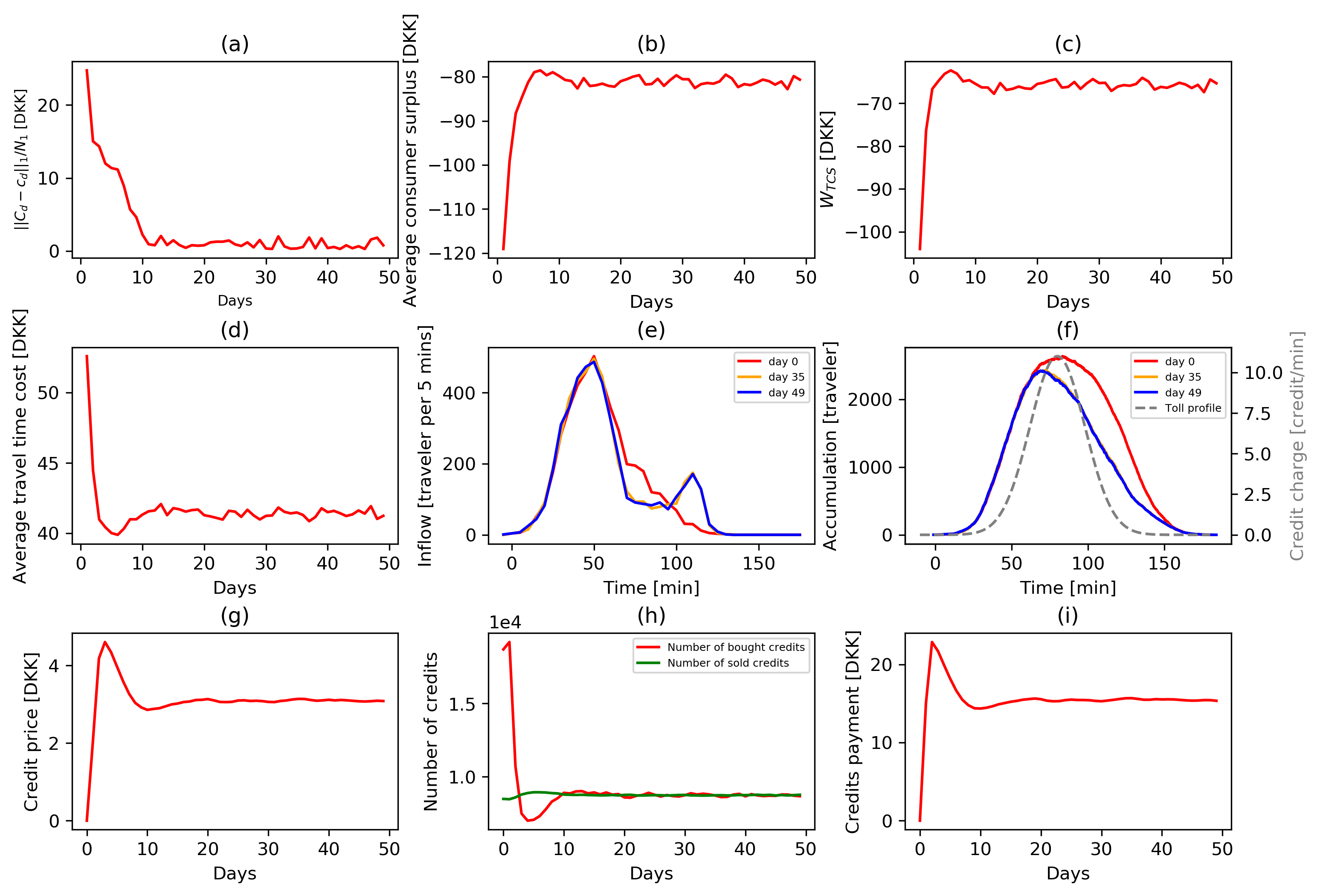}
    \caption{The evolution process of high congestion scenario with TCS}
    \label{TCS4500}
\end{figure}

When applying the same toll profile to the moderate congestion scenario, a similar pattern is observed, the detailed plots are omitted here to conserve space. The day-to-day evolution also reaches a stable state and market equilibrium with a credit price equal to $5.0$ [DKK].

Next, we test the uniqueness of the credit price by setting different initial prices and price adjustment parameter $k$ for the high congestion scenario. Similar patterns are observed in these tests for the moderate congestion scenario. The results are omitted here to be succinct. Figure \ref{Price_e} presents the credit price evolution with different initial prices $p_0=0,~2,~4,~6$ [DKK]. It appears that though the evolution processes are different and start from different initial values, the credit price eventually converges to the same level.

\begin{figure}[H]
    \centering
    \includegraphics[scale=0.75]{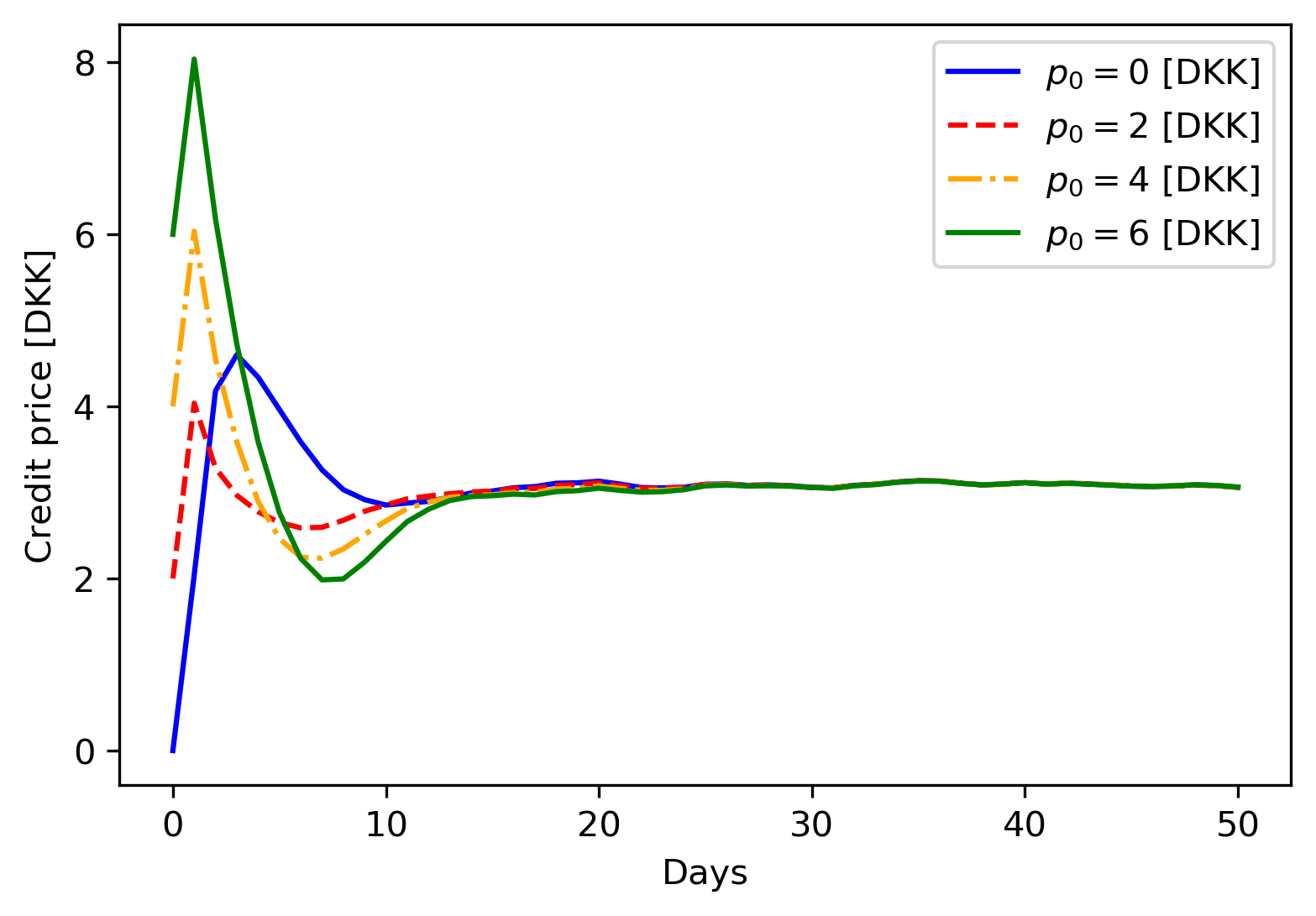}
    \caption{The evolution of credit price with different initial prices}
    \label{Price_e}
\end{figure}

In addition, let $k\in\{0.5 \times 10^{-4},1 \times 10^{-4},2 \times 10^{-4}\}$, we then examine the influence of the price adjustment parameter on price evolution. The results are shown in Figure \ref{Price_k}. When $k$ becomes larger, the credit market shows a greater reaction to the difference between the credit demand and supply, leading to a higher peak value and rapider change in price.

\begin{figure}[H]
    \centering
    \includegraphics[scale=0.75]{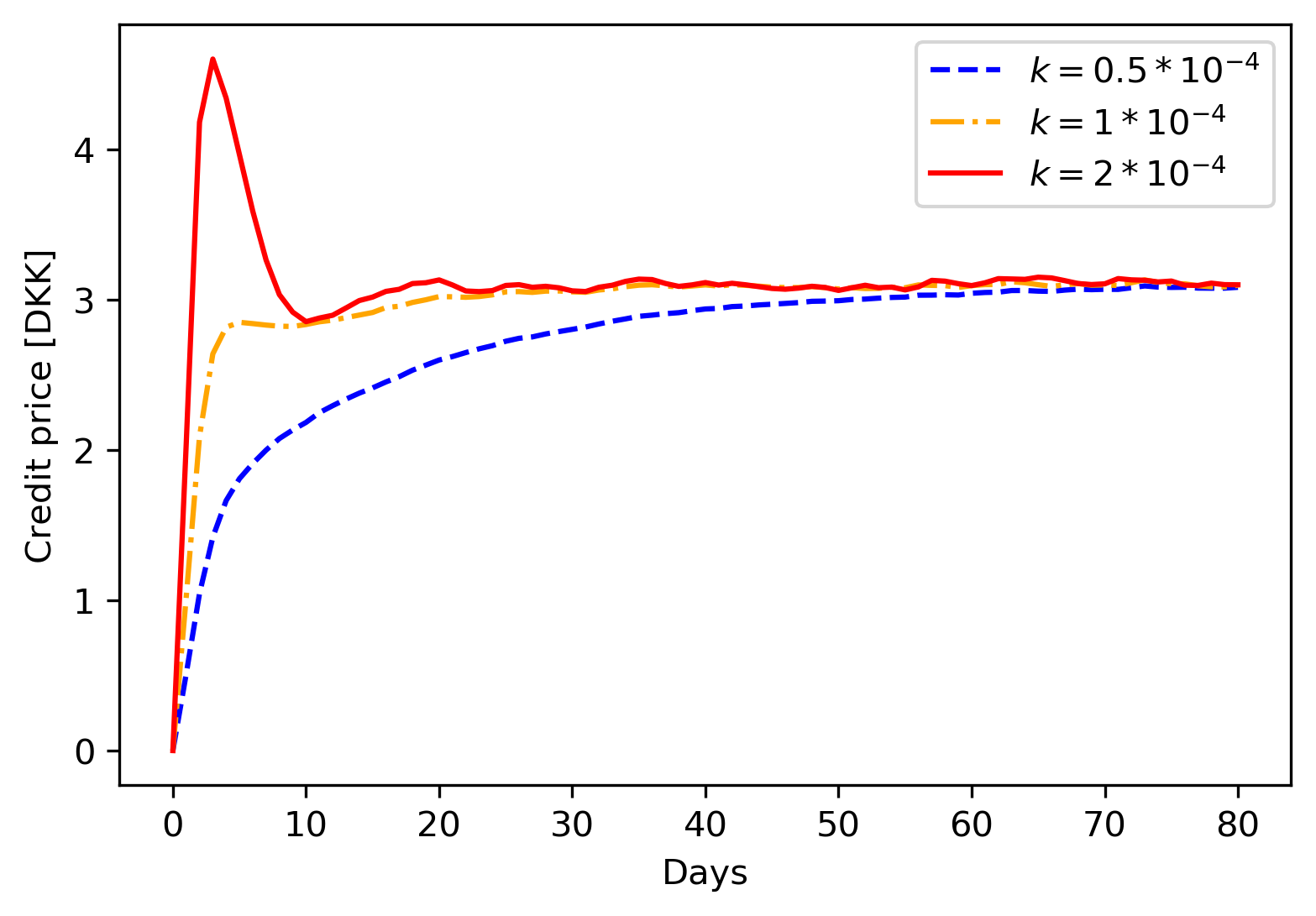}
    \caption{The evolution of credit price with different price adjustment parameters}
    \label{Price_k}
\end{figure}

Moreover, we validate \textbf{Hypothesis 1} by varying the credit endowment. Under the given credit toll profile, the minimum possible endowment $I_{\min}=1.61$ and the $I_{UE}=7.31$. Note the credit endowment is identical among all the travelers and keeps constant across days. Let $I=\{3,4,\dots,7,8\}$, the results are demonstrated in Figure \ref{endowment}. It is clear that the credit price monotonically decreases with $I$ and reaches 0 when $I$ exceeds $I_{UE}=7.31$.

\begin{figure}[H]
    \centering
    \includegraphics[scale=0.75]{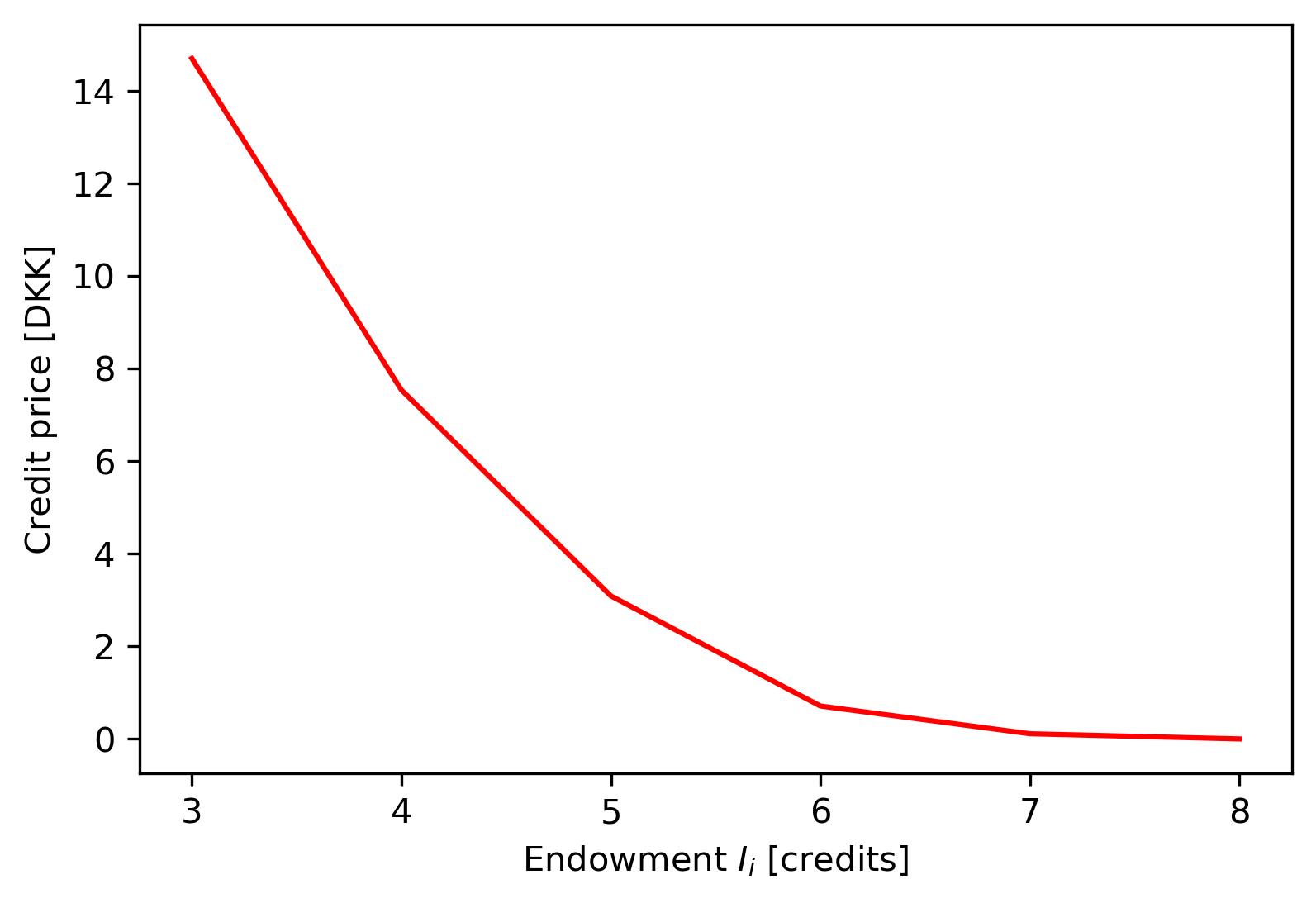}
    \caption{The evolution of credit price with different credit endowment}
    \label{endowment}
\end{figure}

\subsection{Bayesian optimization results}
\label{S5_3}
In this subsection, we present the performance of the TCS by optimizing the toll profile using the developed BO method for both demand scenarios.

The domains of the toll profile function parameters are set as $A\in[5,15]$ (unit: credits), $\xi\in[30,90]$ and $\sigma\in[10,50]$. The initial samples are generated via the LHS consisting of 30 points. For each sample input point, we run the day-to-day simulation and compute the travel time cost and schedule delay cost using the average value of the last 10 days after the equilibrium. The social welfare per capita can therefore be calculated using equation \eqref{sw2}. Figure \ref{optimized} shows the evolution process for the high congestion scenario with optimized toll profile, the parameters of which are $A=5.0$, $\xi=56.0$ and $\sigma=26.1$. It can be seen that the system becomes stable after 30 days, and ends up with a credit price of 10.3 [DKK]. Compared to the no toll case, the departure rate curve is flattened, and the peak accumulation is reduced from 2636 to 1354 [traveler], which overall leads to an improvement in the social welfare, raising from -99.6 [DKK] per capita to -37.0 [DKK] by 62.9\%. We observe from Figure \ref{optimized}(e) and Figure \ref{optimized}(f) that more travelers are departing later under the optimized toll profile compared to the no toll case, in order to avoid high credit tolls. However, due to the highly reduced travel time, there are 68.7\% travelers arriving at their destinations earlier than their desired arrival time under the optimized toll profile, while only 37.1\% in the no toll case.

\begin{figure}[H]
    \centering
    \includegraphics[width=1\textwidth]{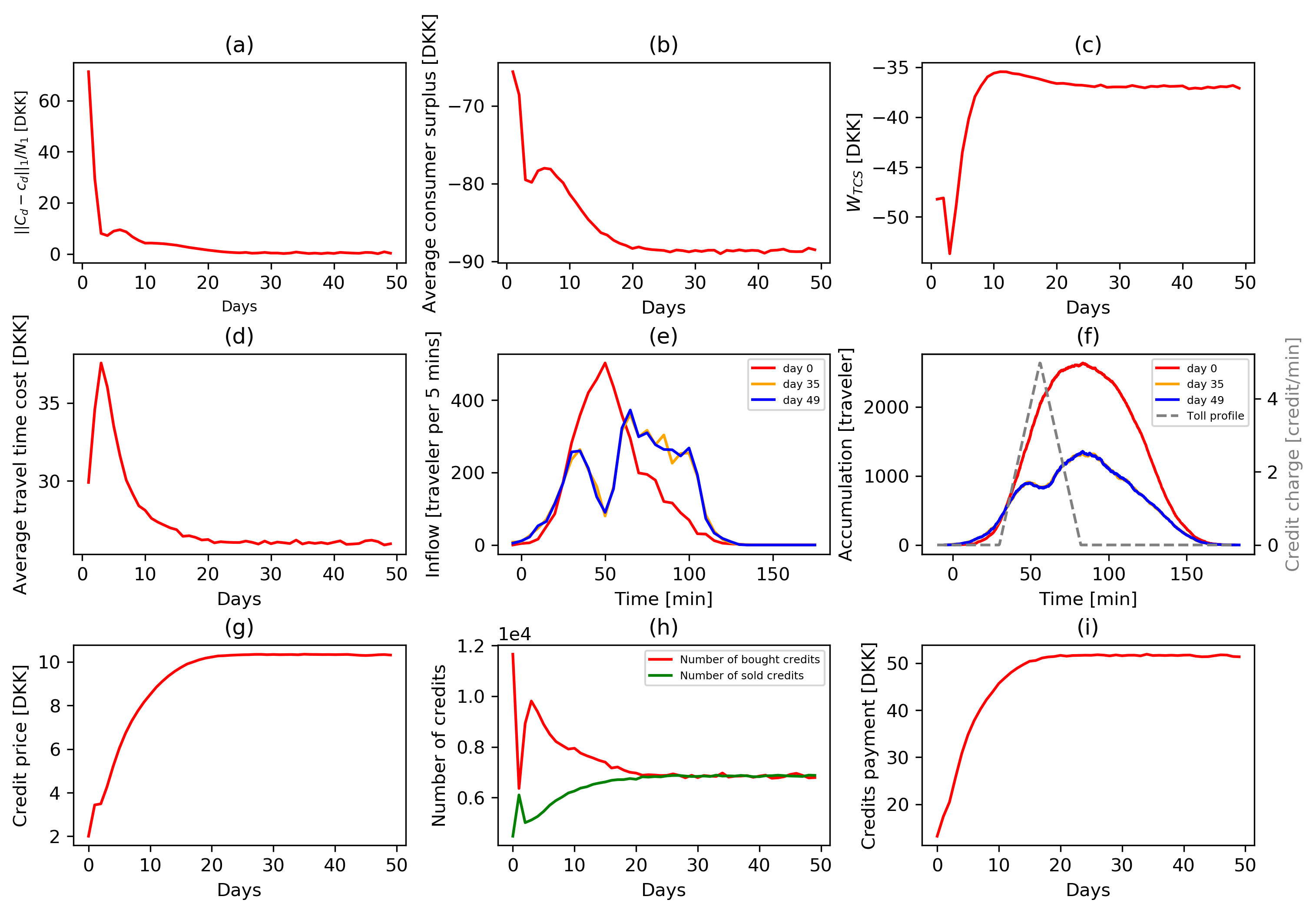}
    \caption{The evolution process of high congestion scenario with optimized TCS}
    \label{optimized}
\end{figure}

Similar patterns are also observed in the moderate congestion scenario, which are shown in Figure \ref{TCS3700best}. The daily average values of measurement variables for the last 10 days after the equilibrium for the no toll case, moderate congestion and high congestion scenario are listed in Table \ref{comparison}, where the second to eighth columns are the daily average monetary travel time cost, schedule delay cost, random utility, consumer surplus, social welfare per capita, toll payment and credit price, respectively. The 'mean' and 'std.dev' rows in Table \ref{comparison} represent the mean values and standard deviation across the last 10 days after the equilibrium for each scenario, respectively. It can be found that, in the moderate congestion scenario, the travel time cost per capita is reduced from 32.4 [DKK] to 27.1 [DKK] by 16.4\%, while the schedule delay cost is increased from 3.7 [DKK] to 4.5 [DKK] by 22.4\%, and overall social welfare per capita is improved by 14.8\%. In the high congestion scenario, the average travel time cost is reduced by 49.4\%, the schedule delay cost is reduced by 71.2\%, and the social welfare per capita is improved by 62.9\%. It is, therefore, concluded that when the congestion is severer, the improvement in terms of social welfare by imposing the optimized TCS is higher.

\begin{figure}[H]
    \centering
    \includegraphics[width=1\textwidth]{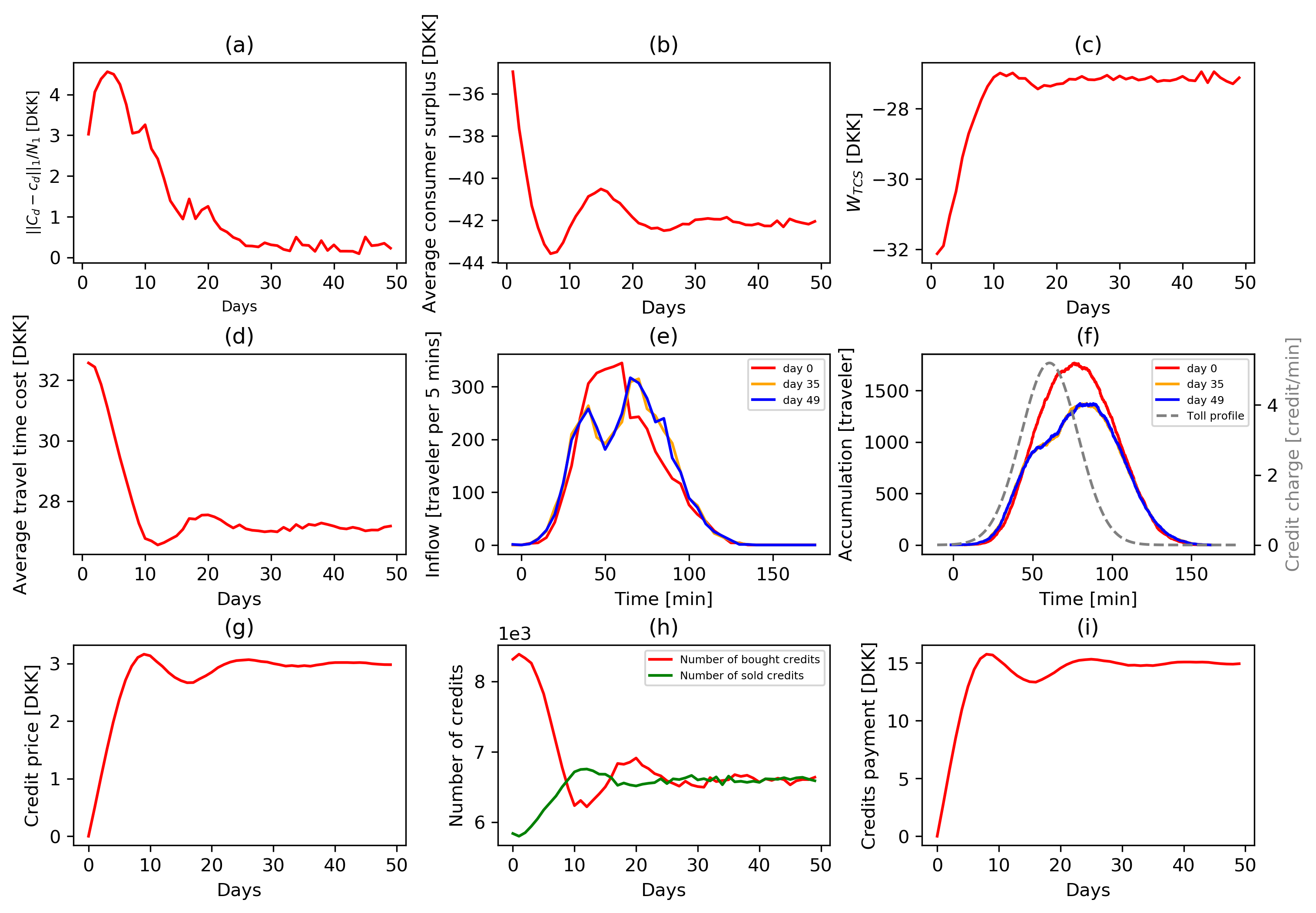}
    \caption{The evolution process of moderate congestion scenario with optimized TCS}
    \label{TCS3700best}
\end{figure}

\begin{table}[H]
  \scriptsize
  \centering
  \caption{Comparisons among no toll case, optimized TCS and optimized CP}\label{comparison}
  \vskip 0.2cm
  \begin{tabular}{p{40pt}p{40pt}p{30pt}p{30pt}p{30pt}p{30pt}p{40pt}p{45pt}}
\hline\noalign{\smallskip}
Unit: [DKK/cap]  & Travel time cost & Schedule delay & Random utility & Consumer surplus & Social welfare & Toll payment & Credit price [DKK]\\
\noalign{\smallskip}\hline\noalign{\smallskip}
\multicolumn{8}{l}{No toll case ($N1$)} \\
Mean & -32.4 & -3.7 & 4.3 & -31.9 & -31.9& -& -\\
Std.dev & 0.1 & 0.05 & 0.07 & 0.18 & 0.18 & -& -\\

\multicolumn{8}{l}{Optimized TCS ($N1$)} \\
Mean & -27.1 & -4.5 & 4.5 & -42.1 & -27.1 & 15.0 & 3.0\\
Std.dev & 0.05 & 0.07 & 0.07 & 0.11 & 0.11 & 0.07 & 0.01\\

\multicolumn{8}{l}{Optimized CP ($N1$)} \\
Mean & -26.3 & -5.2 & 4.4 & -44.4 & -27.1 & 17.3 & - \\
Std.dev & 0.04 & 0.05 & 0.06 & 0.09 &0.08 & 0.05 & -\\

\multicolumn{8}{l}{No toll case ($N2$)} \\
Mean & -51.5 & -51.5 & 3.3 & -99.6 & -99.6 & -& -\\
Std.dev & 0.5 & 1.5 & 0.04 & 2.0 & 2.0 &- & -\\

\multicolumn{8}{l}{Optimized TCS ($N2$)} \\
Mean & -26.1 & -14.8 & 3.9 & -88.6 & -37.0 & 51.6 & 10.3 \\
Std.dev & 0.1 & 0.13 & 0.03 & 0.17 & 0.11 & 0.16 & 0.02\\

\multicolumn{8}{l}{Optimized CP ($N2$)} \\
Mean & -30.6 & -11.1 & 3.9 & -65.1 & -37.9 & 27.3 & -\\
Std.dev & 0.08 & 0.03 & 0.03 & 0.14 & 0.09 & 0.05 & -\\
\hline
\end{tabular}
\end{table}

\subsection{Comparison with time of day pricing}
\label{S5_4}
Under the time of day pricing, travelers' behaviors are simulated based on the same travel behavior model. The only difference is that the toll is set in dollars instead of credits. Thus, the the experienced (or estimated) generalized cost $c_{i,d}^{CP}(t)$ for traveler $i$ on day $d$ departing at time $t$ can be written as follows:
\begin{equation}
\begin{aligned}
\label{exp_cost_cp}
    c_{i,d}^{CP}(t) =& -\theta_i\cdot \Big[T_{i,d}(t)+\delta_i\cdot SDE_i\cdot\big(T_i^*-t-T_{i,d}(t)\big)+\\
    &(1-\delta_i)\cdot SDL_i\cdot\big(t+T_{i,d}(t)-T_i^*\big)\Big]
    -Toll^{CP}(t)\cdot L_i\cdot w\\
    =& -\theta_i\cdot tc_{i,d}(t)-Toll^{CP}(t)\cdot L_i\cdot w
\end{aligned}
\end{equation}
where $Toll^{CP}(t)$ is the toll in dollars at time $t$.

Similar to Section \ref{S_4_1}, we then define the social welfare $W_{CP}$ of congestion pricing, which consists of consumer surplus and regulator revenue:
\begin{equation}
\begin{aligned}
\label{sw3}
    W_{CP} =& CS+RR\\
    =& \frac{1}{N}\sum_{i=1}^N\Big[-\theta_i\cdot tc_{i,d}(t_{i,d}^{dep})- Toll^{CP}(t_{i,d}^{dep})\cdot L_i\cdot w + \epsilon_i(t_{i,d}^{dep})\Big]+\\
    &\frac{1}{N}\sum_{i=1}^N Toll^{CP}(t_{i,d}^{dep})\cdot L_i\cdot w\\
    =& \frac{1}{N}\sum_{i=1}^N\Big[-\theta_i\cdot tc_{i,d}(t_{i,d}^{dep})+\epsilon_i(t_{i,d}^{dep})\Big]\\
    =& \frac{1}{N}\sum_{i=1}^N U'_{i,d}(t_{i,d}^{dep})
\end{aligned}
\end{equation}

The domains of the toll profile function parameter $A$ are slightly different from these before and set as: $A\in[5,20]$ (unit: DKK) and $A\in[5,30]$ (unit: DKK) for moderate and high congestion scenarios, respectively. The BO is used to optimize the toll profile, utilizing LHS sampling method to generate initial points. Figure \ref{CP_opt}(a)-(c) show the evolution process of moderate congestion scenario with optimized toll profile, the parameters of which are $A=18.8$, $\xi=60.7$ and $\sigma=17.9$. Combined with Figure \ref{TCS3700best} and Table \ref{comparison}, time of day pricing reaches the equilibrium of social welfare and flow pattern close to that of TCS case, with a higher toll rate. Figure \ref{CP_opt}(d)-(f) present the evolution process of high congestion scenario, with toll profile parameters $A=30$, $\xi=59.0$ and $\sigma=20.2$. Combined with Figure \ref{optimized} and Table \ref{comparison}, time of day pricing performs worse than TCS, in terms of, for example, social welfare and peak accumulation. As the best $A$ reaches the upper bound of the domain, setting a higher upper bound of $A$ might enable us to find a toll profile which leads to a close performance to the TCS. Nevertheless, it can be concluded that by setting proper toll profiles, which are not adaptive across days, TCS and CP can have the same performance in terms of equilibrium social welfare and flow pattern.

\begin{figure}[H]
    \centering
    \includegraphics[width=1\textwidth]{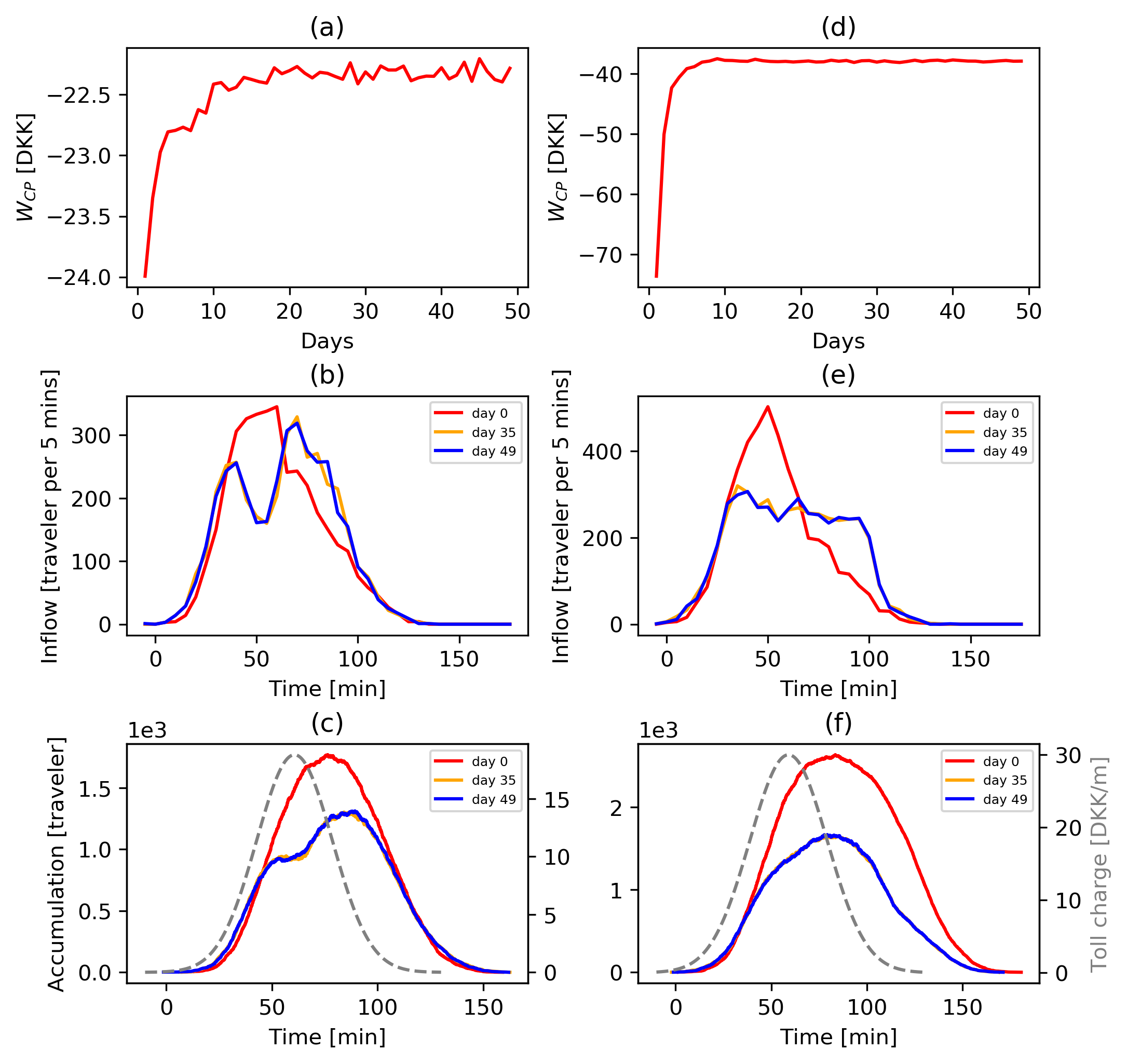}
    \caption{The evolution process of moderate and high congestion scenarios with optimized toll profiles}
    \label{CP_opt}
\end{figure}

\subsection{Comparison of alternative credit toll mechanisms}
\label{S5_5}
\subsubsection{Comparison between different credit toll profiles}
\label{S5_5_1}
Since the performance of the TCS ultimately depends on the choice of the toll profile functional form, here we compare the performance of the above Gaussian profile with two alternatives: a step toll and a triangular toll (T-toll) profile, both inspired by \citep{Zheng2016133} and \citet{daganzo2015distance} respectively, with the caveat that the domain of comparison is limited to the scenario of $N=3700$ and the symmetric profile assumption as before.  For the step toll, there are six parameters including toll charges of five steps and a position parameter indicating the center of the symmetric toll profile. For the T-toll and similarly to the Gaussian case, there are three parameters including the height, length of the base and, again, a position parameter.

Under both the step toll and T-toll, we observe a similar convergence pattern, therefore only the accumulations are shown in
Figure \ref{step_T}. Table \ref{step_T_table} summarizes the detailed information of all the performance measures considered. We relied on 100 Bayesian Optimization iterations for the step-toll case since there are six parameters to optimize while the smaller number of parameters for the two other cases relied on 40 iterations only. It was found that Gaussian toll and T-toll have a similar performance in terms of time related performance measures and welfare. Interestingly, the \textit{CS} shows a difference between the Gaussian and the T-toll, with a higher credit value flow in the latter. Such credit market differences may justify a careful look into efficient market design under more realistic market-related behaviours and equity aspects in TCS related policy decisions. 

Nevertheless, both Gaussian and T-Toll outperformed the step toll. This gap could be reduced under a higher number of steps in the step toll functional form.

\begin{figure}[H]
    \centering
    \includegraphics[width=1\textwidth]{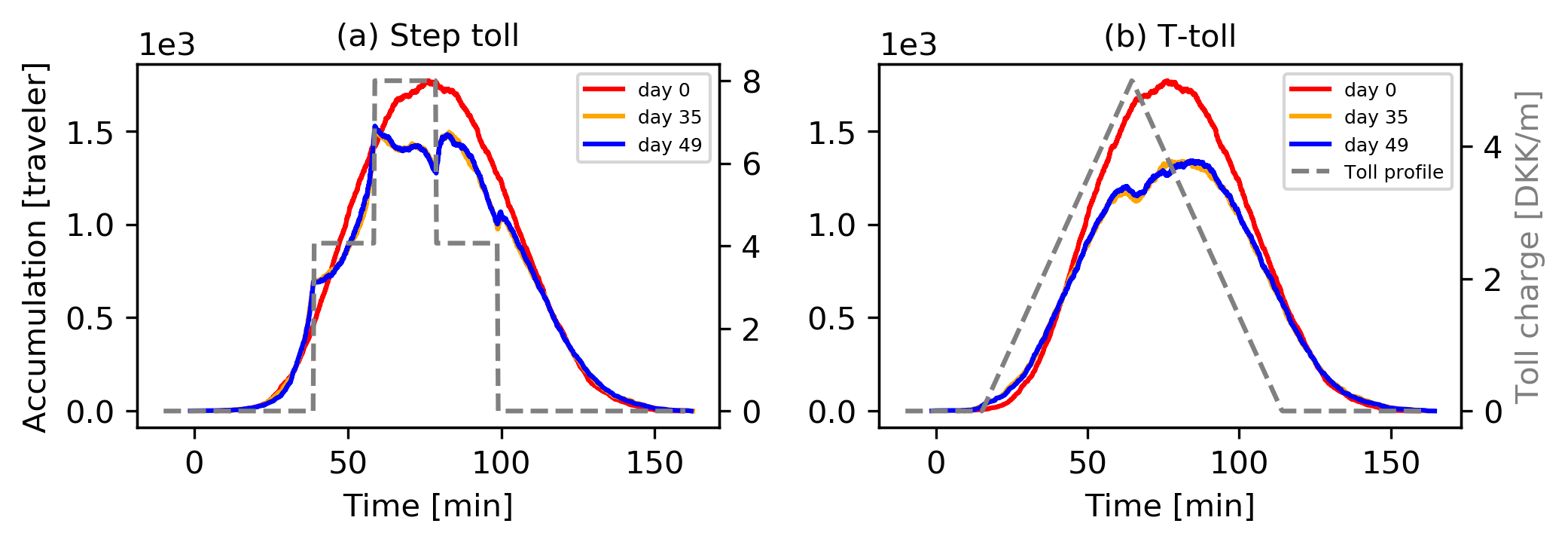}
    \caption{Accumulation for the optimized step toll and T-toll profiles}
    \label{step_T}
\end{figure}

\begin{table}[H]
  \scriptsize
  \centering
  \caption{Comparisons among the optimized Gaussian toll, T-toll and step toll TCS}\label{step_T_table}
  \vskip 0.2cm
  \begin{tabular}{p{40pt}p{40pt}p{30pt}p{30pt}p{30pt}p{30pt}p{40pt}p{45pt}}
\hline\noalign{\smallskip}
Unit: [DKK/cap]  & Travel time cost & Schedule delay & Random utility & Consumer surplus & Social welfare & Toll payment & Credit price [DKK]\\
\noalign{\smallskip}\hline\noalign{\smallskip}
\multicolumn{8}{l}{No toll case ($N1$)} \\
Mean & -32.4 & -3.7 & 4.3 & -31.9 & -31.9& -& -\\
Std.dev & 0.1 & 0.05 & 0.07 & 0.18 & 0.18 & -& -\\

\multicolumn{8}{l}{Optimized Gaussian toll ($N1$)} \\
Mean & -27.1 & -4.5 & 4.5 & -42.1 & -27.1 & 15.0 & 3.0\\
Std.dev & 0.05 & 0.07 & 0.07 & 0.11 & 0.11 & 0.07 & 0.01\\

\multicolumn{8}{l}{Optimized T-toll ($N1$)} \\
Mean & -27.1 & -4.5 & 4.5 & -46.8 & -27.0 & 19.7 & 3.9 \\
Std.dev & 0.04 & 0.04 & 0.07 & 0.10 &0.08 & 0.06 & 0.01\\

\multicolumn{8}{l}{Optimized step toll ($N1$)} \\
Mean & -28.8 & -4.5 & 4.1 & -38.4 & -29.2 & 9.2& 1.8\\
Std.dev & 0.07 & 0.06 & 0.05 & 0.13 & 0.09 &0.12 & 0.02\\
\hline
\end{tabular}
\end{table}

\subsubsection{Comparison between trip-length and travel-time based credit toll mechanisms}
\label{S5_5_2}

One could formulate a credit tolling mechanism function of travel time instead of trip length, allowing a direct accounting of contribution to congestion. Here, we consider a travel time based toll and compare it with the already presented trip length based toll. We simply change the term $p_d\cdot Toll(t)\cdot L_i\cdot w$ in equation \eqref{exp_cost} to $p_d\cdot Toll(t)\cdot T_{i,d}(t)\cdot w'$, where $w'=0.08$ and the other parameters are kept the same. Note that for the trip length based toll, the unit of $Toll(t)$ was [Credit/meter] while for the new travel time based toll, the unit of $Toll(t)$ is [Credit/minute]. Figure \ref{tl_tt} presents the accumulations for the different demand scenarios considered ($N=3700$ and for $N=4500$) for the optimized two types credit tolling mechanisms. The detailed performance measures are summarized in Table \ref{tl_tt_table}.

\begin{figure}[H]
    \centering
    \includegraphics[width=1\textwidth]{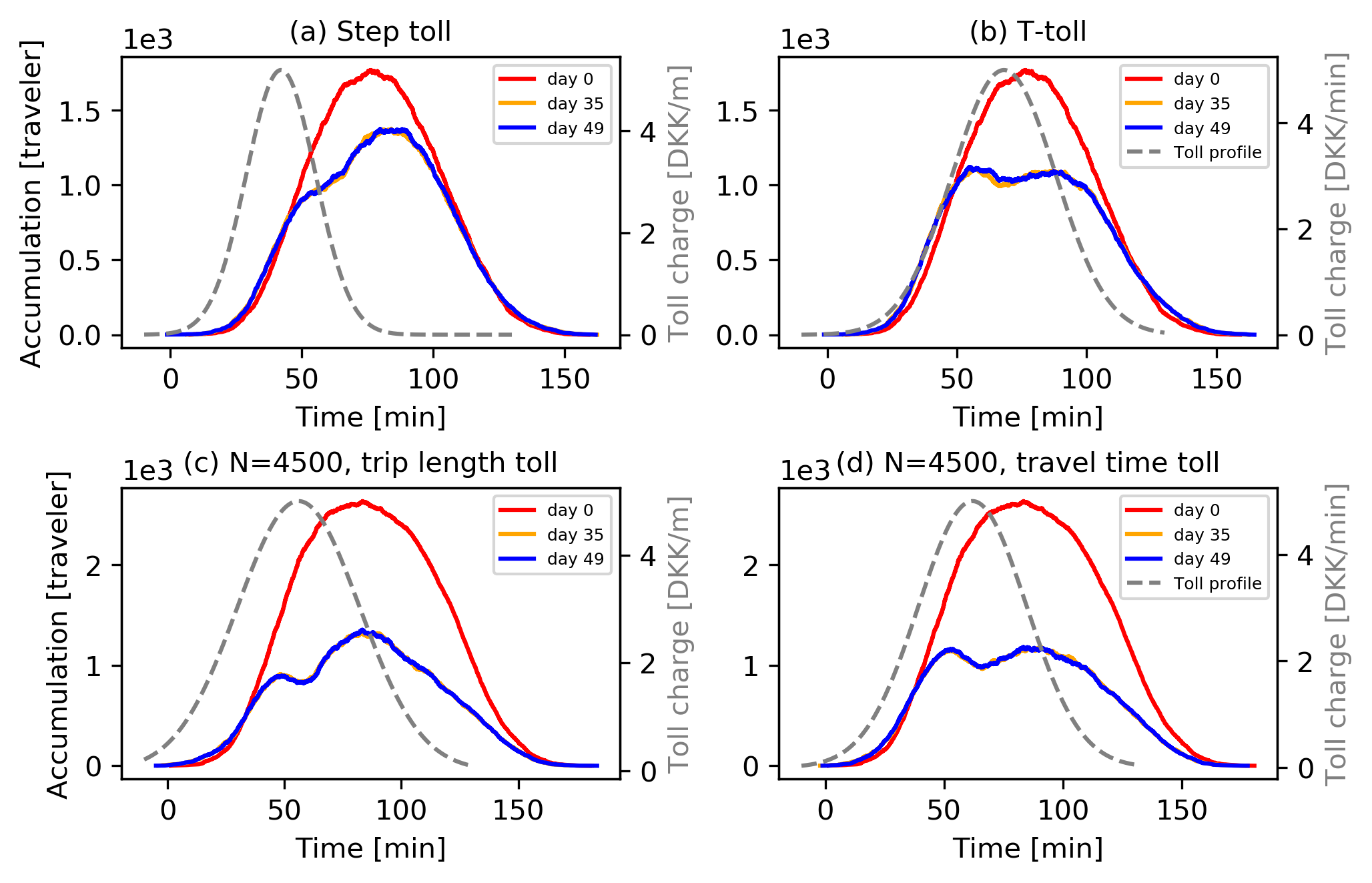}
    \caption{The accumulations under different optimized toll profiles}
    \label{tl_tt}
\end{figure}

\begin{table}[H]
  \scriptsize
  \centering
  \caption{Comparisons between the optimized trip length based TCS and optimized travel time based TCS}\label{tl_tt_table}
  \vskip 0.2cm
  \begin{tabular}{p{40pt}p{40pt}p{30pt}p{30pt}p{30pt}p{30pt}p{40pt}p{45pt}}
\hline\noalign{\smallskip}
Unit: [DKK/cap]  & Travel time cost & Schedule delay & Random utility & Consumer surplus & Social welfare & Toll payment & Credit price [DKK]\\
\noalign{\smallskip}\hline\noalign{\smallskip}
\multicolumn{8}{l}{No toll case ($N1$)} \\
Mean & -32.4 & -3.7 & 4.3 & -31.9 & -31.9& -& -\\
Std.dev & 0.1 & 0.05 & 0.07 & 0.18 & 0.18 & -& -\\

\multicolumn{8}{l}{Optimized trip length based TCS ($N1$)} \\
Mean & -27.1 & -4.5 & 4.5 & -42.1 & -27.1 & 15.0 & 3.0\\
Std.dev & 0.05 & 0.07 & 0.07 & 0.11 & 0.11 & 0.07 & 0.01\\

\multicolumn{8}{l}{Optimized travel time based TCS ($N1$)} \\
Mean & -25.6 & -6.1 & 4.3 & -49.0 & -27.3 & 21.6 & 4.33 \\
Std.dev & 0.04 & 0.06 & 0.06 & 0.10 &0.08 & 0.07 & 0.01\\

\multicolumn{8}{l}{No toll case ($N2$)} \\
Mean & -51.5 & -51.5 & 3.3 & -99.6 & -99.6 & -& -\\
Std.dev & 0.5 & -1.5 & 0.04 & 2.0 & 2.0 &- & -\\

\multicolumn{8}{l}{Optimized trip length based TCS ($N2$)} \\
Mean & -26.1 & -14.8 & 3.9 & -88.6 & -37.0 & 51.6 & 10.3 \\
Std.dev & 0.1 & 0.13 & 0.03 & 0.17 & 0.11 & 0.16 & 0.02\\

\multicolumn{8}{l}{Optimized travel time based TCS ($N2$)} \\
Mean & -26.0 & -14.2 & 3.7 & -78.3 & -36.4 & 41.9 & 8.4\\
Std.dev & 0.08 & 0.11 & 0.03 & 0.25 & 0.06 & 0.28 & 0.02\\
\hline
\end{tabular}
\end{table}

It can be found that both types of toll can converge to an acceptable degree. However, in both demand scenarios, the peak accumulation and departure rate with the travel time based toll are lower than for the trip length based toll. This results in a slightly better welfare performance for the travel time based toll only under the high congested scenario, thanks to small benefits at the schedule day level. Yet, under moderate congested conditions, benefits (disadvantages) in terms of travel time (schedule delay) savings are observed. Indeed, the travel time based toll may reflect the contribution to the congestion more directly. Note that, even with a fixed trip length, when a traveler considers changing the departure time, its associated the credit toll payment will also change. Here, the direct contribution to congestion is taken care by the optimized fixed toll rate and the credit market. Yet, from the traveler's perspective, when a traveler's evaluates departure times, the trip length based toll allows for a clear information on credit payment while the travel time needed for the user's travel time-based toll estimation is uncertain in practice. This falls under information provision and perception modelling research which, while related, is outside of the scope of this manuscript.

\section{Conclusions}
\label{S:6}
This paper proposes a tradable credit scheme (TCS) to manage urban transport network congestion considering the day-to-day evolution of traffic flow. The properties the TCS were examined via both analytical and simulation approaches. 
The properties were analysed in the light of recent generic TCS formulations, namely \cite{bao2019regulating, brands2020tradable}, applied to the case of area-based road traffic control, and extended for heterogeneous trip lengths, i.e. a distance-based tariff instead of access-based tariff. The TCS here at stake relies on a daily fixed credit price, a time-of-day varying tariff charging (or credit toll), a morning commute control policy and heterogeneous decision makers (in terms of choice preferences, trip length, and preferred arrival times).
Meanwhile, a network simulation model is developed to capture the day-to-day evolution of traffic flow. The model is built upon the \textit{trip-based MFD} \citep{arnott2013bathtub, daganzo2015distance} and its efficient implementation proposed in \citep{lamotte2016morning} which allowed us to efficiently study fundamental properties of the TCS. Finally we integrate this overall simulation model that combines the TCS and network simulation model with a Bayesian optimization framework for determining the optimal credit toll charging that maximize the total social welfare.

Analytically, this paper presents the existence of the market and network equilibrium point, the uniqueness of the credit market price, and the associated feasibility conditions. Numerically, the experiments, showcase the  day-to-day  model  properties  and  convergence along with the mobility, network and welfare performances for three comparative polices: no-control case, time-of-day pricing and the proposed TCS.

The numerical results showed good convergence on both the credit price and demonstrated stable network patterns, sustaining the analytical properties on price uniqueness and its inverse proportionality with the endowment. Notably, the proposed TCS improves the social welfare compared to the no-control case and demonstrates promised theoretical performance similarity with the time-of-day pricing. The framework proposed then allowed for a comparison of different credit charging mechanisms. While testing different functional forms for the credit toll profile, the optimized Gaussian-shaped toll had a similar travel and welfare performance to the triangular toll; both outperformed a simple step toll. Moreover, an alternative travel time based credit toll mechanism can perform relatively better than a trip length based toll mechanism in terms of reducing the travel time and enhancing the social welfare in heavy congested scenarios. Yet, while both mechanisms would have to rely on advanced technology for possible implementation, the trip length based toll scheme may have advantages in terms of behavioral uncertainty during the traveller decision making process.

The above developments and findings contribute to the increasing body of knowledge on mobility-related TCS, both in terms of insights into the properties of area-based TCS as well as key modelling and implementations frameworks for the design of future TCS.

In the path for increased knowledge on the feasibility of TCS, the design of TCS markets that accommodate detailed and individual market interactions along with different buying and selling strategies should be analyzed. The buying and selling behaviors are not considered in this study, while they are required for investigating potential market operation models for practical implementation of the TCS, both from a theoretical  \citep{dogterom2017tradable, trinity2020working} and empirical viewpoint \citep{brands2020tradable}. In this study we also kept the TCS settings, including  charging, endowment, and credit price, constant within a day. Nevertheless, it is acknowledged that adaptive credit charging, sporadic endowment and quantity control interventions by the regulator and real-time / within-day credit price adjustment may bring the TCS closer to efficient operations, especially under the non-recurrent conditions of real transportation system. Yet, detailed simulation and behavioural experiments approaches may again be required to overcome the common simplifying assumptions for analytical tractability. Nonetheless, the aforementioned findings of this paper bring insights into possible modelling techniques to include in the design and real-time operations of practice ready area-wide TCS. Finally, the consideration of additional and combined choice dimensions in future TCS efficiency analysis (such as mode, route, departure time and trip cancellation) is currently lacking in the current literature \citep{akamatsu2017tradable}, yet it is in much need for bringing TCS closer to practice .

\section*{Acknowledgements}
\label{S:acknowledgements}
This research was partially carried out under the NEMESYS project funded by the DTU-NTU (Nanyang Technical University) Alliance and the Trinity project funded by the U.S. National Science Foundation (ID:CMMI-1917891). We also thank the anonymous reviewers for valuable comments and suggestions.

\appendix

\section{Proof of Theorem 1}
\label{P_thm1}
This proof begins with 

\begin{lemma}\label{lemma1}
$C_{i}$ and $p$ solve
\begin{equation}\label{a1}
    Q(p,Z)=0
\end{equation}
iff
\begin{equation}\label{a2}
    p=[p+\rho\cdot Z]_+
\end{equation}
where $\rho$ is a constant larger than 0.
\end{lemma}

\begin{proof}
Sufficiency: from (\ref{funQ}), Eq. (\ref{a1}) holds iff
\begin{equation}\label{a3}
    p\cdot Z=0,\quad p\geq 0,\quad Z\leq 0
\end{equation}
If $C_{i}$ and $p$ satisfy (\ref{a3}), then either
\begin{equation}\label{a4}
    p=0,\quad Z\leq 0
\end{equation}
or
\begin{equation}\label{a5}
    p\geq 0,\quad Z=0
\end{equation}
holds. Obviously, (\ref{a4}) and (\ref{a5}) satisfy (\ref{a2})

Necessity: if $C_{i}$ and $p$ satisfy (\ref{a2}), then either
\begin{equation}\label{a6}
\left\{
\begin{aligned}
     & p=0,\\
     & \rho\cdot Z\leq 0
\end{aligned}
\right.
\end{equation}
or
\begin{equation}\label{a7}
\left\{
\begin{aligned}
     & p>0,\\
     & Z= 0
\end{aligned}
\right.
\end{equation}
holds. And as (\ref{a6}) and (\ref{a7}) satisfy (\ref{a3}), then (\ref{a1}) holds. Therefore, Lemma \ref{lemma1} holds. \qed
\end{proof}

\begin{proof1}
By Lemma \ref{lemma1}, the equilibrium condition (\ref{equilibrium2}) is equivalent to
\begin{equation}\label{a8}
\left\{
\begin{aligned}
     & C_{i}(t_i)=\theta_i\cdot tc_{i}(t_i)+p\cdot Toll(t_i)\cdot L_i\cdot w\\
     & p=[p+\rho\cdot Z]_+\\
\end{aligned}
\right.
\end{equation}
Then, we introduce the fixed point theorem from \cite{khamsi2011introduction}:
\begin{theorem1}\label{thma1}
Let $\Omega$ be a bounded closed convex subset of $\mathbb{R}^m$ and let $\bm{g}:\Omega\rightarrow\Omega$ be continuous. Then $\bm{g}$ has a fixed point.
\end{theorem1}
According to \textbf{Algorithm 1}, it is obvious that for any $t_i\in [t_{i,0}^{dep}-\tau\cdot\Delta t,t_{i,0}^{dep}+\tau\cdot\Delta t]$, there exists a $tc_i(t_i)$ and $\lim_{t\rightarrow t_i}tc_i(t)=tc_i(t_i)$ holds, thus travel cost $tc_{i}(t_i)$ is non-negative and continuous. Since $p$ and $Toll(t_i)$ are also non-negative and continuous with regard to $t_i$, we can conclude that $C_{i}(t_i)$ is non-negative and continuous. Let $\bm{t}=[t_1,\dots,t_m]^T$, where $t_i\in [t_{i,0}^{dep}-\tau\cdot\Delta t,t_{i,0}^{dep}+\tau\cdot\Delta t]$, $\bm{C}(\bm{t},p)=[C_1(t_1),\dots,C_m(t_m)]^T$, $\bm{\theta}=[\theta_1,\dots,\theta_m]^T$, $\bm{tc}=[tc_1(t_1),\dots,tc_m(t_m)]^T$, $\bm{Toll}=[Toll(t_1),\dots,Toll(t_m)]^T$, and $\bm{L}=[L_1,\dots,L_m]^T$. Then there exists $\bm{y}=(y_1,\dots,y_m)^T\geq \bm{0}$, such that $\bm{0}\leq \bm{C}(\bm{t},p)\leq \bm{y}$. Therefore, a compact and convex set can be defined as $\Omega_{\bm{C}}=[0,y_1]\times\cdots\times[0,y_m]$, then for all $\bm{C}\in\Omega_{\bm{C}}$, $\bm{\theta\circ tc}+p\cdot w\cdot \bm{Toll}\circ\bm{ L}\in \Omega_{\bm{C}}$.

Denote $Z_d=\sum_i Toll(t_i)\cdot L_i\cdot w-I_{i,d}\cdot N$, then
\begin{equation}\label{a9}
\begin{aligned}
\lim_{p\rightarrow\infty}Z_d(\bm{C},p)
=&\lim_{p\rightarrow\infty}\sum_i Toll(t_i)\cdot L_i\cdot w-I_{i,d}\cdot N\\
=&(I_{\min}-I)N< 0,\quad \forall \bm{C}\in\Omega_{\bm{C}}
\end{aligned}
\end{equation}
For some $\bm{\hat{C}}\in\Omega_{\bm{C}}$, if 
\begin{equation}\label{a10}
    Z(\bm{\hat{C}},p)\leq 0, \forall p\geq 0
\end{equation}
then $[p+\rho Z(\bm{\hat{C}},p)]_+\leq p,\forall p\geq 0$. Let us define $\Omega_{\bm{C}}^-$ as the set of $\bm{\hat{C}}\in\Omega_{\bm{C}}$ satisfying condition (\ref{a10}) and define $\Omega_{\bm{C}}^+=\Omega_{\bm{C}}\backslash\Omega_{\bm{C}}^-$. Then there exists $\bm{\hat{C}}\in\Omega_{\bm{C}}^+$, such that $Z(\bm{\hat{C}},p)>0$ for some $p\geq 0$. According to (\ref{a9}) and (\ref{a10}), there exists $\bar{p}\geq 0$, such that $Z(\bm{\hat{C}},\bar{p})\leq 0,~\forall p\geq \bar{p}\Rightarrow [p+\rho Z(\bm{\hat{C}},p)]_+\leq p,\forall p\geq \bar{p}$. Let $p^+=\max_{p\leq\bar{p}}[p+\rho Z(\bm{\hat{C}},p)]_+$, then $\forall p\in [0,p^+],[p+\rho Z(\bm{\hat{C}},p)]_+\leq p^+$. Therefore $\forall p\in\Omega_p\triangleq[0,\max_{\bm{\hat{C}}\in\Omega_{\bm{C}}^+}p^+]$ and $\bm{C}\in\Omega_{\bm{C}}$, $[p+\rho Z(\bm{\hat{C}},p)]_+\in \Omega_{p}$.

Based on the analysis above, $\Omega_{\bm{C}}\times\Omega_{p}$ is compact and convex. Since travel cost $\bm{tc}(\cdot)$, $[\cdot]$ and $Toll(\cdot)$ are continuous, then by Theorem A\ref{thma1}, (\ref{a8}) has at least one fixed point, implying that there exists at least one one equilibrium point of the proposed dynamic system. \qed
\end{proof1}

\section{Proof of Theorem 2}
\label{P_thm2}
\begin{proof2}
Assume there are two equilibrium credit prices $p_1$ and $p_2$, then by (\ref{price}) and (\ref{funQ}), we have
\begin{equation}\label{funQ_proof}
\begin{aligned}
    &p_1\big[\sum_i Toll(t_{i,p_1}^{dep})\cdot L_i\cdot w-I_{i,d}\cdot N\big]=0,~p_1\geq0,\\
    &p_2\big[\sum_i Toll(t_{i,p_2}^{dep})\cdot L_i\cdot w-I_{i,d}\cdot N\big]=0,~p_2\geq0,\\
    &\sum_i Toll(t_{i,p_1}^{dep})\cdot L_i\cdot w-I_{i,d}\cdot N \leq0,\\
    &\sum_i Toll(t_{i,p_2}^{dep})\cdot L_i\cdot w-I_{i,d}\cdot N \leq0
\end{aligned}    
\end{equation}
Thus,
\begin{equation}\label{price_proof}
\begin{aligned}
    &(p_1-p_2)\Big(\sum_i Toll(t_{i,p_1}^{dep})\cdot L_i\cdot w-\sum_i Toll(t_{i,p_2}^{dep})\cdot L_i\cdot w\Big)\\
    &=p_1 \sum_i Toll(t_{i,p_1}^{dep})\cdot L_i\cdot w-p_1 \sum_i Toll(t_{i,p_2}^{dep})\cdot L_i\cdot w \\
    &+ p_2 \sum_i Toll(t_{i,p_2}^{dep})\cdot L_i\cdot w - p_2\sum_i Toll(t_{i,p_1}^{dep})\cdot L_i\cdot w\\
    &=p_1\big[I_{i,d}\cdot N-\sum_i Toll(t_{i,p_2}^{dep})\cdot L_i\cdot w\big]+p_2\big[I_{i,d}\cdot N-\sum_i Toll(t_{i,p_1}^{dep})\cdot L_i\cdot w\big]\\
    &\geq 0
\end{aligned}    
\end{equation}
By (\ref{p_condition}), the equality in (\ref{price_proof}) holds if and only if $p_1=p_2$. Hence, the equilibrium credit price is unique. \qed
\end{proof2}
%% References



%% New version of the num-names style
\bibliographystyle{elsarticle-harv}
\bibliography{ref.bib}

\begin{thebibliography}{54}
\expandafter\ifx\csname natexlab\endcsname\relax\def\natexlab#1{#1}\fi
\providecommand{\url}[1]{\texttt{#1}}
\providecommand{\href}[2]{#2}
\providecommand{\path}[1]{#1}
\providecommand{\DOIprefix}{doi:}
\providecommand{\ArXivprefix}{arXiv:}
\providecommand{\URLprefix}{URL: }
\providecommand{\Pubmedprefix}{pmid:}
\providecommand{\doi}[1]{\href{http://dx.doi.org/#1}{\path{#1}}}
\providecommand{\Pubmed}[1]{\href{pmid:#1}{\path{#1}}}
\providecommand{\bibinfo}[2]{#2}
\ifx\xfnm\relax \def\xfnm[#1]{\unskip,\space#1}\fi
%Type = Article
\bibitem[{Akamatsu and Wada(2017)}]{akamatsu2017tradable}
\bibinfo{author}{Akamatsu, T.}, \bibinfo{author}{Wada, K.},
  \bibinfo{year}{2017}.
\newblock \bibinfo{title}{Tradable network permits: A new scheme for the most
  efficient use of network capacity}.
\newblock \bibinfo{journal}{Transportation Research Part C: Emerging
  Technologies} \bibinfo{volume}{79}, \bibinfo{pages}{178--195}.
%Type = Article
\bibitem[{Amirgholy and Gao(2017)}]{Amirgholy2017215}
\bibinfo{author}{Amirgholy, M.}, \bibinfo{author}{Gao, H.O.},
  \bibinfo{year}{2017}.
\newblock \bibinfo{title}{Modeling the dynamics of congestion in large urban
  networks using the macroscopic fundamental diagram: User equilibrium, system
  optimum, and pricing strategies}.
\newblock \bibinfo{journal}{Transportation Research Part B: Methodological}
  \bibinfo{volume}{104}, \bibinfo{pages}{215 -- 237}.
%Type = Article
\bibitem[{Arnott(2013)}]{arnott2013bathtub}
\bibinfo{author}{Arnott, R.}, \bibinfo{year}{2013}.
\newblock \bibinfo{title}{A bathtub model of downtown traffic congestion}.
\newblock \bibinfo{journal}{Journal of Urban Economics} \bibinfo{volume}{76},
  \bibinfo{pages}{110--121}.
%Type = Article
\bibitem[{Bao et~al.(2019)Bao, Verhoef and Koster}]{bao2019regulating}
\bibinfo{author}{Bao, Y.}, \bibinfo{author}{Verhoef, E.T.},
  \bibinfo{author}{Koster, P.}, \bibinfo{year}{2019}.
\newblock \bibinfo{title}{Regulating dynamic congestion externalities with
  tradable credit schemes: Does a unique equilibrium exist?}
\newblock \bibinfo{journal}{Transportation Research Part B: Methodological}
  \bibinfo{volume}{127}, \bibinfo{pages}{225--236}.
%Type = Article
\bibitem[{Ben-Akiva et~al.(1984)Ben-Akiva, Cyna and De~Palma}]{ben1984dynamic}
\bibinfo{author}{Ben-Akiva, M.}, \bibinfo{author}{Cyna, M.},
  \bibinfo{author}{De~Palma, A.}, \bibinfo{year}{1984}.
\newblock \bibinfo{title}{Dynamic model of peak period congestion}.
\newblock \bibinfo{journal}{Transportation Research Part B: Methodological}
  \bibinfo{volume}{18}, \bibinfo{pages}{339--355}.
%Type = Book
\bibitem[{Ben-Akiva et~al.(1985)Ben-Akiva, Lerman and Lerman}]{ben1985discrete}
\bibinfo{author}{Ben-Akiva, M.E.}, \bibinfo{author}{Lerman, S.R.},
  \bibinfo{author}{Lerman, S.R.}, \bibinfo{year}{1985}.
\newblock \bibinfo{title}{Discrete choice analysis: theory and application to
  travel demand}. volume~\bibinfo{volume}{9}.
\newblock \bibinfo{publisher}{MIT press}.
%Type = Article
\bibitem[{Brands et~al.(2020)Brands, Verhoef, Knockaert and
  Koster}]{brands2020tradable}
\bibinfo{author}{Brands, D.K.}, \bibinfo{author}{Verhoef, E.T.},
  \bibinfo{author}{Knockaert, J.}, \bibinfo{author}{Koster, P.R.},
  \bibinfo{year}{2020}.
\newblock \bibinfo{title}{Tradable permits to manage urban mobility: market
  design and experimental implementation}.
\newblock \bibinfo{journal}{Transportation Research Part A: Policy and
  Practice} \bibinfo{volume}{137}, \bibinfo{pages}{34--46}.
%Type = Article
\bibitem[{Cantarella and Cascetta(1995)}]{cantarella1995dynamic}
\bibinfo{author}{Cantarella, G.E.}, \bibinfo{author}{Cascetta, E.},
  \bibinfo{year}{1995}.
\newblock \bibinfo{title}{Dynamic processes and equilibrium in transportation
  networks: towards a unifying theory}.
\newblock \bibinfo{journal}{Transportation Science} \bibinfo{volume}{29},
  \bibinfo{pages}{305--329}.
%Type = Article
\bibitem[{Cantarella et~al.(2015)Cantarella, Velon{\`a} and
  Watling}]{cantarella2015day}
\bibinfo{author}{Cantarella, G.E.}, \bibinfo{author}{Velon{\`a}, P.},
  \bibinfo{author}{Watling, D.P.}, \bibinfo{year}{2015}.
\newblock \bibinfo{title}{Day-to-day dynamics \& equilibrium stability in a
  two-mode transport system with responsive bus operator strategies}.
\newblock \bibinfo{journal}{Networks and Spatial Economics}
  \bibinfo{volume}{15}, \bibinfo{pages}{485--506}.
%Type = Techreport
\bibitem[{Chen et~al.(2020)Chen, Seshadri, Lima~Azevedo, Akkinepally, Liu,
  Araldo, Jiang and Ben-Akiva}]{trinity2020working}
\bibinfo{author}{Chen, S.}, \bibinfo{author}{Seshadri, R.},
  \bibinfo{author}{Lima~Azevedo, C.}, \bibinfo{author}{Akkinepally, A.P.},
  \bibinfo{author}{Liu, R.}, \bibinfo{author}{Araldo, A.},
  \bibinfo{author}{Jiang, Y.}, \bibinfo{author}{Ben-Akiva, M.},
  \bibinfo{year}{2020}.
\newblock \bibinfo{title}{Analysis and Design of Markets for Tradable Mobility
  Credit Schemes}.
\newblock \bibinfo{type}{Working Paper, available on arXiv}.
%Type = Article
\bibitem[{Chu(1995)}]{CHU1995324}
\bibinfo{author}{Chu, X.}, \bibinfo{year}{1995}.
\newblock \bibinfo{title}{Endogenous trip scheduling: The henderson approach
  reformulated and compared with the vickrey approach}.
\newblock \bibinfo{journal}{Journal of Urban Economics} \bibinfo{volume}{37},
  \bibinfo{pages}{324 -- 343}.
%Type = Article
\bibitem[{Daganzo(2007)}]{daganzo2007urban}
\bibinfo{author}{Daganzo, C.F.}, \bibinfo{year}{2007}.
\newblock \bibinfo{title}{Urban gridlock: Macroscopic modeling and mitigation
  approaches}.
\newblock \bibinfo{journal}{Transportation Research Part B: Methodological}
  \bibinfo{volume}{41}, \bibinfo{pages}{49--62}.
%Type = Article
\bibitem[{Daganzo and Lehe(2015)}]{daganzo2015distance}
\bibinfo{author}{Daganzo, C.F.}, \bibinfo{author}{Lehe, L.J.},
  \bibinfo{year}{2015}.
\newblock \bibinfo{title}{Distance-dependent congestion pricing for downtown
  zones}.
\newblock \bibinfo{journal}{Transportation Research Part B: Methodological}
  \bibinfo{volume}{75}, \bibinfo{pages}{89--99}.
%Type = Article
\bibitem[{Dogterom et~al.(2017)Dogterom, Ettema and
  Dijst}]{dogterom2017tradable}
\bibinfo{author}{Dogterom, N.}, \bibinfo{author}{Ettema, D.},
  \bibinfo{author}{Dijst, M.}, \bibinfo{year}{2017}.
\newblock \bibinfo{title}{Tradable credits for managing car travel: a review of
  empirical research and relevant behavioural approaches}.
\newblock \bibinfo{journal}{Transport Reviews} \bibinfo{volume}{37},
  \bibinfo{pages}{322--343}.
%Type = Article
\bibitem[{Fan and Jiang(2013)}]{fan2013tradable}
\bibinfo{author}{Fan, W.}, \bibinfo{author}{Jiang, X.}, \bibinfo{year}{2013}.
\newblock \bibinfo{title}{Tradable mobility permits in roadway capacity
  allocation: Review and appraisal}.
\newblock \bibinfo{journal}{Transport Policy} \bibinfo{volume}{30},
  \bibinfo{pages}{132--142}.
%Type = Article
\bibitem[{Fosgerau(2015)}]{fosgerau2015congestion}
\bibinfo{author}{Fosgerau, M.}, \bibinfo{year}{2015}.
\newblock \bibinfo{title}{Congestion in the bathtub}.
\newblock \bibinfo{journal}{Economics of Transportation} \bibinfo{volume}{4},
  \bibinfo{pages}{241--255}.
%Type = Article
\bibitem[{Fosgerau et~al.(2007)Fosgerau, Hjorth and
  Lyk-Jensen}]{fosgerau2007danish}
\bibinfo{author}{Fosgerau, M.}, \bibinfo{author}{Hjorth, K.},
  \bibinfo{author}{Lyk-Jensen, S.V.}, \bibinfo{year}{2007}.
\newblock \bibinfo{title}{The danish value of time study} .
%Type = Article
\bibitem[{Fosgerau and Small(2013)}]{fosgerau2013hypercongestion}
\bibinfo{author}{Fosgerau, M.}, \bibinfo{author}{Small, K.A.},
  \bibinfo{year}{2013}.
\newblock \bibinfo{title}{Hypercongestion in downtown metropolis}.
\newblock \bibinfo{journal}{Journal of Urban Economics} \bibinfo{volume}{76},
  \bibinfo{pages}{122--134}.
%Type = Article
\bibitem[{Frazier(2018)}]{frazier2018tutorial}
\bibinfo{author}{Frazier, P.I.}, \bibinfo{year}{2018}.
\newblock \bibinfo{title}{A tutorial on bayesian optimization}.
\newblock \bibinfo{journal}{arXiv preprint arXiv:1807.02811} .
%Type = Inproceedings
\bibitem[{Geroliminis et~al.(2007)Geroliminis, Daganzo
  et~al.}]{geroliminis2007macroscopic}
\bibinfo{author}{Geroliminis, N.}, \bibinfo{author}{Daganzo, C.F.}, et~al.,
  \bibinfo{year}{2007}.
\newblock \bibinfo{title}{Macroscopic modeling of traffic in cities}, in:
  \bibinfo{booktitle}{Transportation Research Board 86th Annual Meeting},
  \bibinfo{organization}{No. 07-0413}.
%Type = Article
\bibitem[{Gonzales and Daganzo(2012)}]{gonzales2012morning}
\bibinfo{author}{Gonzales, E.J.}, \bibinfo{author}{Daganzo, C.F.},
  \bibinfo{year}{2012}.
\newblock \bibinfo{title}{Morning commute with competing modes and distributed
  demand: user equilibrium, system optimum, and pricing}.
\newblock \bibinfo{journal}{Transportation Research Part B: Methodological}
  \bibinfo{volume}{46}, \bibinfo{pages}{1519--1534}.
%Type = Article
\bibitem[{Grant-Muller and Xu(2014)}]{grant2014role}
\bibinfo{author}{Grant-Muller, S.}, \bibinfo{author}{Xu, M.},
  \bibinfo{year}{2014}.
\newblock \bibinfo{title}{The role of tradable credit schemes in road traffic
  congestion management}.
\newblock \bibinfo{journal}{Transport Reviews} \bibinfo{volume}{34},
  \bibinfo{pages}{128--149}.
%Type = Article
\bibitem[{Guo et~al.(2016)Guo, Yang, Huang and Tan}]{guo2016day}
\bibinfo{author}{Guo, R.Y.}, \bibinfo{author}{Yang, H.},
  \bibinfo{author}{Huang, H.J.}, \bibinfo{author}{Tan, Z.},
  \bibinfo{year}{2016}.
\newblock \bibinfo{title}{Day-to-day flow dynamics and congestion control}.
\newblock \bibinfo{journal}{Transportation Science} \bibinfo{volume}{50},
  \bibinfo{pages}{982--997}.
%Type = Article
\bibitem[{Han et~al.(2017)Han, Wang, Lo, Zhu and Cai}]{han2017discrete}
\bibinfo{author}{Han, L.}, \bibinfo{author}{Wang, D.Z.}, \bibinfo{author}{Lo,
  H.K.}, \bibinfo{author}{Zhu, C.}, \bibinfo{author}{Cai, X.},
  \bibinfo{year}{2017}.
\newblock \bibinfo{title}{Discrete-time day-to-day dynamic congestion pricing
  scheme considering multiple equilibria}.
\newblock \bibinfo{journal}{Transportation Research Part B: Methodological}
  \bibinfo{volume}{104}, \bibinfo{pages}{1--16}.
%Type = Article
\bibitem[{Horowitz(1984)}]{horowitz1984stability}
\bibinfo{author}{Horowitz, J.L.}, \bibinfo{year}{1984}.
\newblock \bibinfo{title}{The stability of stochastic equilibrium in a two-link
  transportation network}.
\newblock \bibinfo{journal}{Transportation Research Part B: Methodological}
  \bibinfo{volume}{18}, \bibinfo{pages}{13--28}.
%Type = Article
\bibitem[{Johnston et~al.(1995)Johnston, Lund and Craig}]{johnston1995capacity}
\bibinfo{author}{Johnston, R.A.}, \bibinfo{author}{Lund, J.R.},
  \bibinfo{author}{Craig, P.P.}, \bibinfo{year}{1995}.
\newblock \bibinfo{title}{Capacity-allocation methods for reducing urban
  traffic congestion}.
\newblock \bibinfo{journal}{Journal of Transportation Engineering}
  \bibinfo{volume}{121}, \bibinfo{pages}{27--39}.
%Type = Book
\bibitem[{Khamsi and Kirk(2011)}]{khamsi2011introduction}
\bibinfo{author}{Khamsi, M.A.}, \bibinfo{author}{Kirk, W.A.},
  \bibinfo{year}{2011}.
\newblock \bibinfo{title}{An introduction to metric spaces and fixed point
  theory}. volume~\bibinfo{volume}{53}.
\newblock \bibinfo{publisher}{John Wiley \& Sons}.
%Type = Inproceedings
\bibitem[{Lamotte and Geroliminis(2015)}]{lamotte2015dynamic}
\bibinfo{author}{Lamotte, R.}, \bibinfo{author}{Geroliminis, N.},
  \bibinfo{year}{2015}.
\newblock \bibinfo{title}{Dynamic traffic modeling: Approximating the
  equilibrium for peak periods in urban areas}, in: \bibinfo{booktitle}{15th
  Swiss Transport Research Conference, Monte Verità, Ascona, April 15 – 17}.
%Type = Techreport
\bibitem[{Lamotte and Geroliminis(2016)}]{lamotte2016morning}
\bibinfo{author}{Lamotte, R.}, \bibinfo{author}{Geroliminis, N.},
  \bibinfo{year}{2016}.
\newblock \bibinfo{title}{The morning commute in urban areas: Insights from
  theory and simulation}.
\newblock \bibinfo{type}{Technical Report}. Transportation Research Board 95th
  Annual Meeting.
%Type = Article
\bibitem[{Lamotte and Geroliminis(2018)}]{lamotte2018morning}
\bibinfo{author}{Lamotte, R.}, \bibinfo{author}{Geroliminis, N.},
  \bibinfo{year}{2018}.
\newblock \bibinfo{title}{The morning commute in urban areas with heterogeneous
  trip lengths}.
\newblock \bibinfo{journal}{Transportation Research Part B: Methodological}
  \bibinfo{volume}{117}, \bibinfo{pages}{794--810}.
%Type = Article
\bibitem[{Leclercq et~al.(2015)Leclercq, Parzani, Knoop, Amourette and
  Hoogendoorn}]{leclercq2015macroscopic}
\bibinfo{author}{Leclercq, L.}, \bibinfo{author}{Parzani, C.},
  \bibinfo{author}{Knoop, V.L.}, \bibinfo{author}{Amourette, J.},
  \bibinfo{author}{Hoogendoorn, S.P.}, \bibinfo{year}{2015}.
\newblock \bibinfo{title}{Macroscopic traffic dynamics with heterogeneous route
  patterns}.
\newblock \bibinfo{journal}{Transportation Research Part C: Emerging
  Technologies} \bibinfo{volume}{59}, \bibinfo{pages}{292--307}.
%Type = Article
\bibitem[{Lessan and Fu(2019)}]{lessan2019credit}
\bibinfo{author}{Lessan, J.}, \bibinfo{author}{Fu, L.}, \bibinfo{year}{2019}.
\newblock \bibinfo{title}{Credit-and permit-based travel demand management
  state-of-the-art methodological advances}.
\newblock \bibinfo{journal}{Transportmetrica A: Transport Science} ,
  \bibinfo{pages}{1--24}.
%Type = Article
\bibitem[{Lindsey(2006)}]{lindsey2006economists}
\bibinfo{author}{Lindsey, R.}, \bibinfo{year}{2006}.
\newblock \bibinfo{title}{Do economists reach a conclusion?}
\newblock \bibinfo{journal}{Econ Journal Watch} \bibinfo{volume}{3},
  \bibinfo{pages}{292--379}.
%Type = Article
\bibitem[{Liu et~al.(2020)Liu, Jiang and Azevedo}]{liu2020bayesian}
\bibinfo{author}{Liu, R.}, \bibinfo{author}{Jiang, Y.},
  \bibinfo{author}{Azevedo, C.L.}, \bibinfo{year}{2020}.
\newblock \bibinfo{title}{Bayesian optimization of area-based road pricing}.
\newblock \bibinfo{journal}{arXiv preprint arXiv:2012.11047} .
%Type = Article
\bibitem[{Liu and Geroliminis(2017)}]{liu2017doubly}
\bibinfo{author}{Liu, W.}, \bibinfo{author}{Geroliminis, N.},
  \bibinfo{year}{2017}.
\newblock \bibinfo{title}{Doubly dynamics for multi-modal networks with
  park-and-ride and adaptive pricing}.
\newblock \bibinfo{journal}{Transportation Research Part B: Methodological}
  \bibinfo{volume}{102}, \bibinfo{pages}{162--179}.
%Type = Article
\bibitem[{Mariotte et~al.(2017)Mariotte, Leclercq and
  Laval}]{mariotte2017macroscopic}
\bibinfo{author}{Mariotte, G.}, \bibinfo{author}{Leclercq, L.},
  \bibinfo{author}{Laval, J.A.}, \bibinfo{year}{2017}.
\newblock \bibinfo{title}{Macroscopic urban dynamics: Analytical and numerical
  comparisons of existing models}.
\newblock \bibinfo{journal}{Transportation Research Part B: Methodological}
  \bibinfo{volume}{101}, \bibinfo{pages}{245--267}.
%Type = Book
\bibitem[{Mat{\'e}rn(2013)}]{matern2013spatial}
\bibinfo{author}{Mat{\'e}rn, B.}, \bibinfo{year}{2013}.
\newblock \bibinfo{title}{Spatial variation}. volume~\bibinfo{volume}{36}.
\newblock \bibinfo{publisher}{Springer Science \& Business Media}.
%Type = Article
\bibitem[{McKay et~al.(2000)McKay, Beckman and Conover}]{mckay2000comparison}
\bibinfo{author}{McKay, M.D.}, \bibinfo{author}{Beckman, R.J.},
  \bibinfo{author}{Conover, W.J.}, \bibinfo{year}{2000}.
\newblock \bibinfo{title}{A comparison of three methods for selecting values of
  input variables in the analysis of output from a computer code}.
\newblock \bibinfo{journal}{Technometrics} \bibinfo{volume}{42},
  \bibinfo{pages}{55--61}.
%Type = Article
\bibitem[{Miralinaghi and Peeta(2016)}]{miralinaghi2016multi}
\bibinfo{author}{Miralinaghi, M.}, \bibinfo{author}{Peeta, S.},
  \bibinfo{year}{2016}.
\newblock \bibinfo{title}{Multi-period equilibrium modeling planning framework
  for tradable credit schemes}.
\newblock \bibinfo{journal}{Transportation Research Part E: Logistics and
  Transportation Review} \bibinfo{volume}{93}, \bibinfo{pages}{177--198}.
%Type = Article
\bibitem[{de~Palma and Lindsey(2020)}]{de2020tradable}
\bibinfo{author}{de~Palma, A.}, \bibinfo{author}{Lindsey, R.},
  \bibinfo{year}{2020}.
\newblock \bibinfo{title}{Tradable permit schemes for congestible facilities
  with uncertain supply and demand}.
\newblock \bibinfo{journal}{Economics of Transportation} \bibinfo{volume}{21},
  \bibinfo{pages}{100149}.
%Type = Article
\bibitem[{de~Palma et~al.(2018)de~Palma, Proost, Seshadri and
  Ben-Akiva}]{de2018congestion}
\bibinfo{author}{de~Palma, A.}, \bibinfo{author}{Proost, S.},
  \bibinfo{author}{Seshadri, R.}, \bibinfo{author}{Ben-Akiva, M.},
  \bibinfo{year}{2018}.
\newblock \bibinfo{title}{Congestion tolling-dollars versus tokens: A
  comparative analysis}.
\newblock \bibinfo{journal}{Transportation Research Part B: Methodological}
  \bibinfo{volume}{108}, \bibinfo{pages}{261--280}.
%Type = Book
\bibitem[{Pigou(2013)}]{pigou2013economics}
\bibinfo{author}{Pigou, A.C.}, \bibinfo{year}{2013}.
\newblock \bibinfo{title}{The economics of welfare}.
\newblock \bibinfo{publisher}{Palgrave Macmillan}.
%Type = Book
\bibitem[{Sheffi(1984)}]{sheffi1984urban}
\bibinfo{author}{Sheffi, Y.}, \bibinfo{year}{1984}.
\newblock \bibinfo{title}{Urban Transportation Networks: Equilibrium Analysis
  with Mathematical Programming Methods}.
\newblock \bibinfo{publisher}{Prentice-Hall}.
%Type = Article
\bibitem[{Srinivas et~al.(2009)Srinivas, Krause, Kakade and
  Seeger}]{srinivas2009gaussian}
\bibinfo{author}{Srinivas, N.}, \bibinfo{author}{Krause, A.},
  \bibinfo{author}{Kakade, S.M.}, \bibinfo{author}{Seeger, M.},
  \bibinfo{year}{2009}.
\newblock \bibinfo{title}{Gaussian process optimization in the bandit setting:
  No regret and experimental design}.
\newblock \bibinfo{journal}{arXiv preprint arXiv:0912.3995} .
%Type = Article
\bibitem[{Tan et~al.(2015)Tan, Yang and Guo}]{tan2015dynamic}
\bibinfo{author}{Tan, Z.}, \bibinfo{author}{Yang, H.}, \bibinfo{author}{Guo,
  R.Y.}, \bibinfo{year}{2015}.
\newblock \bibinfo{title}{Dynamic congestion pricing with day-to-day flow
  evolution and user heterogeneity}.
\newblock \bibinfo{journal}{Transportation Research Part C: Emerging
  Technologies} \bibinfo{volume}{61}, \bibinfo{pages}{87--105}.
%Type = Article
\bibitem[{Verhoef et~al.(1997)Verhoef, Nijkamp and
  Rietveld}]{verhoef1997tradeable}
\bibinfo{author}{Verhoef, E.}, \bibinfo{author}{Nijkamp, P.},
  \bibinfo{author}{Rietveld, P.}, \bibinfo{year}{1997}.
\newblock \bibinfo{title}{Tradeable permits: their potential in the regulation
  of road transport externalities}.
\newblock \bibinfo{journal}{Environment and Planning B: Planning and Design}
  \bibinfo{volume}{24}, \bibinfo{pages}{527--548}.
%Type = Article
\bibitem[{Vickrey(1969)}]{vickrey1969congestion}
\bibinfo{author}{Vickrey, W.S.}, \bibinfo{year}{1969}.
\newblock \bibinfo{title}{Congestion theory and transport investment}.
\newblock \bibinfo{journal}{The American Economic Review} \bibinfo{volume}{59},
  \bibinfo{pages}{251--260}.
%Type = Book
\bibitem[{Williams and Rasmussen(2006)}]{williams2006gaussian}
\bibinfo{author}{Williams, C.K.}, \bibinfo{author}{Rasmussen, C.E.},
  \bibinfo{year}{2006}.
\newblock \bibinfo{title}{Gaussian processes for machine learning}.
  volume~\bibinfo{volume}{2}.
\newblock \bibinfo{publisher}{MIT press Cambridge, MA}.
%Type = Article
\bibitem[{Wu et~al.(2012)Wu, Yin, Lawphongpanich and Yang}]{wu2012design}
\bibinfo{author}{Wu, D.}, \bibinfo{author}{Yin, Y.},
  \bibinfo{author}{Lawphongpanich, S.}, \bibinfo{author}{Yang, H.},
  \bibinfo{year}{2012}.
\newblock \bibinfo{title}{Design of more equitable congestion pricing and
  tradable credit schemes for multimodal transportation networks}.
\newblock \bibinfo{journal}{Transportation Research Part B: Methodological}
  \bibinfo{volume}{46}, \bibinfo{pages}{1273--1287}.
%Type = Article
\bibitem[{Xiao et~al.(2013)Xiao, Qian and Zhang}]{xiao2013managing}
\bibinfo{author}{Xiao, F.}, \bibinfo{author}{Qian, Z.S.},
  \bibinfo{author}{Zhang, H.M.}, \bibinfo{year}{2013}.
\newblock \bibinfo{title}{Managing bottleneck congestion with tradable
  credits}.
\newblock \bibinfo{journal}{Transportation Research Part B: Methodological}
  \bibinfo{volume}{56}, \bibinfo{pages}{1--14}.
%Type = Article
\bibitem[{Yang and Wang(2011)}]{yang2011managing}
\bibinfo{author}{Yang, H.}, \bibinfo{author}{Wang, X.}, \bibinfo{year}{2011}.
\newblock \bibinfo{title}{Managing network mobility with tradable credits}.
\newblock \bibinfo{journal}{Transportation Research Part B: Methodological}
  \bibinfo{volume}{45}, \bibinfo{pages}{580--594}.
%Type = Article
\bibitem[{Ye and Yang(2013)}]{ye2013continuous}
\bibinfo{author}{Ye, H.}, \bibinfo{author}{Yang, H.}, \bibinfo{year}{2013}.
\newblock \bibinfo{title}{Continuous price and flow dynamics of tradable
  mobility credits}.
\newblock \bibinfo{journal}{Transportation Research Part B: Methodological}
  \bibinfo{volume}{57}, \bibinfo{pages}{436--450}.
%Type = Article
\bibitem[{Yildirimoglu and Ramezani(2020)}]{yildirimoglu2020demand}
\bibinfo{author}{Yildirimoglu, M.}, \bibinfo{author}{Ramezani, M.},
  \bibinfo{year}{2020}.
\newblock \bibinfo{title}{Demand management with limited cooperation among
  travellers: A doubly dynamic approach}.
\newblock \bibinfo{journal}{Transportation Research Part B: Methodological}
  \bibinfo{volume}{132}, \bibinfo{pages}{267--284}.
%Type = Article
\bibitem[{Zheng et~al.(2016)Zheng, Rérat and Geroliminis}]{Zheng2016133}
\bibinfo{author}{Zheng, N.}, \bibinfo{author}{Rérat, G.},
  \bibinfo{author}{Geroliminis, N.}, \bibinfo{year}{2016}.
\newblock \bibinfo{title}{Time-dependent area-based pricing for multimodal
  systems with heterogeneous users in an agent-based environment}.
\newblock \bibinfo{journal}{Transportation Research Part C: Emerging
  Technologies} \bibinfo{volume}{62}, \bibinfo{pages}{133 -- 148}.

\end{thebibliography}






% \end{thebibliography}

\end{document}